\documentclass[runningheads]{file}

\usepackage[T2A]{fontenc}

\usepackage{graphicx}
\usepackage{tikz}
\usepackage{color, amsmath, amssymb}
\usepackage{lineno}
   
\date{}

\DeclareMathOperator{\longg}{long}
\DeclareMathOperator{\short}{short}
\DeclareMathOperator{\zeros}{zeros}
\DeclareMathOperator{\fix}{fix}

\newcommand{\qedSymC}{\hfill$\lhd$}

\begin{document}

\title{Recovery of cyclic words by their subwords} 

\institute{}

\author{Sergey Luchinin\inst{1}, Svetlana Puzynina\inst{1}, Micha\"el Rao\inst{2}}

\authorrunning{}

\institute{Saint Petersburg State University, Russia 
\and  ENS de Lyon, CNRS, Universit\'e de Lyon, France\\
\email{\{serg2000lrambler.ru,s.puzynina\}@gmail.com,michael.rao@ens-lyon.fr}
}

\maketitle  

\begin{abstract}
The problem of reconstructing words from their subwords involves determining the minimum amount of information needed, such as multisets of scattered subwords of a specific length or the frequency of scattered subwords from a given set, in order to uniquely identify a word. In this paper we show that a cyclic word on a binary alphabet can be reconstructed by its scattered subwords of length $\frac34n+4$, and for each $n$ one can find two cyclic words of length $n$ which have the same set of scattered subwords of length $\frac34n-\frac32$.   
\end{abstract}

\section{Introduction}

The problem of reconstruction of words by their subwords is an important topic in combinatorics of words and related fields of mathematics and theoretical computer science. There are many  variants of this problem: reconstruction of normal or cyclic words, from the sets of scattered or contiguous subwords, with or without multiplicities, restricting to a particular family of words and so on. It has been known for more than 40 years that one can reconstruct a word of length $n$ from the set of all its scattered subwords of length $\lfloor\frac{n-1}{2}\rfloor$, and this bound is optimal \cite{Lo83}. The same bound holds for contiguous subwords with multiplicities, while without multiplicities it is $(n-1)$ \cite{Manuch}.
 
The question of reconstruction of words of length $n$ by their scattered subwords of length $k$ with multiplicities turned out to be more complicated. One of the first results related to this question states that for words over a binary alphabet and for $k \geqslant \lfloor \frac{n}2\rfloor$ we can uniquely reconstruct the word, and for $k < \log_2(n)$ we cannot do it \cite{MMSSS91}. One of the best upper bounds of this problem with subword length $\lfloor\frac{16}{7} \sqrt{n}\rfloor$ follows from a result related to polynomials introduced in \cite{BEK99}. One of the best lower bounds is $k = \Omega(e^{\log^{\frac12}(n)})$ \cite{DS03}.

Other results about words reconstruction include reconstruction from the number of occurrences of scattered factors of some special form. For example, a word $w \in \{a,b\}^*$ can be reconstructed from the number of occurrences of at most $\min(|w|_a, |w|_b)+1$ scattered factors of the form $a^ib$, where $|w|_a$ is the number of occurrences of the letter $a$ in $w$ \cite{FLM2021}. A word can also be reconstructed from the number of occurrences as scattered subwords of $O(\ell^2\sqrt{n\log(n)})$ words we have chosen, where $\ell$ is the cardinality of alphabet \cite{RR23} (this improves the results from \cite{FLM2021}). Besides that, there are some results about reconstruction of words of special types. For example, reconstruction of DNA-words has been studied in \cite{SL04}, and reconstruction of words that are cyclic shifts of some fixed word has been explored in \cite{L12}.

In this paper we consider the problem of reconstructing a cyclic word by its scattered subwords without multiplicities:

\medskip

\textbf{Problem:} Consider a cyclic binary word $w$ of length $n$ and the set of its scattered cyclic subwords of length at most $k$ (without multiplicities). For which $k$ can we always recover $w$?

\medskip

Computational results giving the minimal length $k$ of subwords which allows to recover cyclic words of length $n$ for $n \leqslant 32$ are provided in Table 1.

\begin{table}[h]
\centering
\caption{Computational results for small values of $n$}
\begin{tabular}{|c|c|c|c|c|c|c|c|c|c|c|c|c|c|c|c|c|c|c|c|c|c|c|c|c|c|c|c|c|c|c|c|c|}
\hline $n$ & 1 & 2 & 3 & 4 & 5 & 6 & 7 & 8 & 9 & 10 & 11 & 12 & 13
& 14
& 15 & 16 & 17 & 18 & 19 & 20 & 21 & 22 & 23 & 24 & 25 & 26 & 27 & 28 & 29 & 30 & 31 & 32\\
\hline $k$ & 1 & 1 & 2 & 4 & 4 & 6 & 6 & 7 & 8 & 8 & 9 & 10 & 11 & 11 & 12 & 12 & 14 & 14 & 14 & 15 & 17 & 16 & 17 & 18 & 20 & 19 & 20 & 21 & 23 & 22 & 23 & 24\\
\hline
\end{tabular}
\end{table}

The main result of this paper is the following theorem:

\begin{theorem} \label{th:main}
For any two distinct cyclic binary words $u$ and $v$ of length $n$ there exists a word $w$ of length at most $\frac34n+4$ which is a subword of exactly one of the words $u$ and $v$. 
\end{theorem}


We also show that for each $n \geqslant 7$ there are pairs of words for which sets of subwords of length at most $\frac34n-\frac32$ are equal (see Proposition~\ref{pr:lower_bound}). So, we found a lower bound $k \geqslant \frac34n-\frac{3}{2}$ and an upper bound $k\leqslant\frac34n+4$. In other words, we calculated $k$ modulo some constant which is at most $5$, and therefore our bound is almost sharp.

The paper is organized as follows. In Section 2, we introduce necessary definitions and notation used throughout the text. In Section 3, we prove some auxiliary propositions that are needed for the proof of the theorem. Besides that, in this section we provide a lower bound for $k$. In Section 4, we prove the main theorem. The section is divided into two subsections: in Subsection 4.1  we introduce the notation used throughout the section and the general structure of the proof, and Subsection 4.2 contains the proof of the theorem divided into five lemmas.   

\section{Definitions and notation}

Let $\Sigma$ be an alphabet. In the paper, we assume $\Sigma$ to be binary. A \textit{word} is a finite or infinite sequence of symbols from $\Sigma$, and $\Sigma^*$ denotes the set of all finite words. 

Two words $x$ and $y$ are said to be \textit{conjugate} if there exist words $u$ and $v$ such that $x=uv$ and $y=vu$. We define a \textit{cyclic word} as an equivalence class of the conjugacy relation on $\Sigma^*$. Thus, if $w \in \Sigma^*$, then the cyclic word represented by $w$ is the set $ \{vu \in \Sigma^* \mid u,v\in\Sigma^*, uv = w \}$. For the rest of this paper, when referring to a cyclic word, we write a representative of this class, slightly abusing the notion to avoid cluttering the text. For a finite or a cyclic word its \textit{length} is the number of letters in it.

For a cyclic word $w$ with a representative $w_1\cdots w_n$, a \textit{subword} of $w$ is a cyclic word with a representative of the form $w_{i_1}\cdots w_{i_k}$, where $1\leqslant i_1 < \ldots < i_k\leqslant n$. 
A factor of some conjugate of $w_1\cdots w_n$ of the form $0^+$ (resp., $1^+$) continued to the left and to the right with 1 (resp., 0) is called a {\it block} of 0's (resp., 1's).

For two cyclic words $u$ and $v$, we say that a cyclic word $w$ is a \textit{distinguishing} subword if it is a subword of only one of the words $u$ and $v$. Using this notion, Theorem~\ref{th:main} can be reformulated as follows: any two distinct cyclic binary words of length $n$ have a distinguishing subword of length at most $\frac{3}{4} n + 4$.

We say that a cyclic word $w$ is \textit{periodic} if $w = (1^{\alpha_1}0^{\beta_1}1^{\alpha_2}0^{\beta_2}\cdots 1^{\alpha_s}0^{\beta_s})^r$ for some $\alpha_i \geqslant  1$, $\beta_i \geqslant  1$, $s\geq 1$ and $r \geqslant  2$.

\bigskip

Let $w$ be a cyclic word of length $n$. Throughout the paper, we make use of the following notation:

\begin{itemize}


\item $n_{0,w}$ and $n_{1,w}$ are the numbers of 0's and 1's in $w$.

\item $2l_w$ is the total  number of blocks in $w$, $l_w$ blocks of 0's and $l_w$ blocks of 1's (here we consider a representative of $w$ in which the first and the last letters are distinct).

\item $x_w$ is the length of the longest block of 0's in $w$.

\item Blocks $0^{x_w}$ are called {\it big blocks}, other blocks are called {\it small blocks}.

\item $y_w$ is the length of the longest block in $w$ which is smaller than $x_w$ (we might have $y_w = 0$).

\item $n_{\longg, w}$ is the number of 0's in big blocks $0^{x_w}$ and $n_{\short, w} = n_{0,w} - n_{\longg, w}$ is the number of 0's in small blocks.

\item $w_{\longg}$ is the subword of $w$ which contains all 1's of $w$ and all 0's from big blocks of 0's (all blocks $0^{x_w}$).

\item $w_{\short}$ is the subword of $w$ which contains all 1's of $w$ and all 0's from small blocks of 0's (all blocks of 0's except for the blocks $0^{x_w}$).

\item $w_{\zeros}$ is the subword of $w$ which contains all 1's of $w$ and one $0$ from each block of 0's.
\end{itemize}

For the rest of this paper, we omit the subscript $w$ when no confusion arises. The notation $n_0, n_1$ is justified by the following. In the beginning of the next section we prove Proposition~\ref{prop n_0, n_1, l} stating that cyclic words $u$ and $v$ with the same set of subwords of length $\frac{3}{4} n + 4$ have the same numbers of 0's, 1's and also the same number of blocks. So, for the rest of the paper $n_0, n_1$ and $l$ are fixed.

\begin{example}
Let $w = 10^3101^20^31^20^2$. Then $n_1 = 6, n_0 = 9, x = 3, 
n_{\longg} = 6, n_{\short} = 3, w_{\longg} = 1^30^31^30^3, w_{\zeros} = 10101^201^20$. 
\end{example} 

For the proof of the main result, we need to treat words of certain specific forms separately. We hence introduce the following definition:

\begin{definition}
Let $w$ be a word with $l_w\geqslant  2$. We say that $w$ is \emph{special} if it is of one of the following three types:

\begin{itemize}
\item  \emph{first type}: 
$$w = (0^t1^m)^{l_w} $$ 
\item \emph{second type}:  
$$w = (0^t1^m)^{l_w-1} 0^t1^{2m} $$
\item \emph{third type}: 
$$w = (0^t1^m)^i 0^t1^{2m} (0^t1^m)^{l_w-i-2}  0^t1^{2m}$$
for some positive integers $t$, $m$ and $i$. \end{itemize} \end{definition}

So, a word of the first type has $l_w$ blocks $1^m$, a word of the second type has $l_w-1$ blocks $1^m$ and one block $1^{2m}$, a word of the third type has $l_w-2$ blocks $1^m$ and two blocks $1^{2m}$, and the number of blocks of 0's is at least two in each case. 


For a special word $w$, we let $dist(w)$ denote the length of the shortest block of 1's, i.e., the minimum distance between blocks of 0's. Note that the value of $dist(w)$ is equal to the number $m$ from the definition of a special word. 

When working with cyclic words, sometimes we  need to index their letters, either just one letter or both letters. For example, for a cyclic word generated by $00101$ we could either index 1's as $001_101_2$, or as $001_201_1$. We remark that it is not exactly the same as choosing a representative from the conjugacy class (the difference comes up in periodic cyclic words).

\begin{definition} \label{def:2}
    Let $u$ and $v$ be cyclic words with indexed 1's such that $n_{0,u}=n_{0,v}$ and $n_{1,u}=n_{1,v}$. We then define a \emph{1-overlay} of the word $u$ on the word $v$ as a bijection between indexed 1's in $u$ and in $v$ which, for some integer $i$, translates each $1_j$ in $u$ to $1_{j+i}$ in $v$ (indices are taken modulo~$n_1$). 
\end{definition}

An example of a $1$-overlay is provided on Fig.~\ref{fig:overlay}.

\begin{center}
\begin{figure}
\begin{tikzpicture}

\draw (3,0) circle (1);
\draw (4,0) node [right]  {$1_2$}
(3,1) node [above] {$1_1$}
(3,-1) node [below] {$1_3$}
(2,0) node [left]  {$1_4$}
(3.7,0.7) node [above right] {$0^3$}
(2.3,0.7) node [above left] {$0^2$}
(3.7,-0.7) node [below right] {}
(2.3,-0.7) node [below left] {$0^2$};

\draw (7,0) circle (1);
\draw (8,0) node [right]  {\color{blue}$1_2$}
(7,1) node [above] {\color{blue}$1_1$}
(7,-1) node [below] {\color{blue}$1_3$}
(6,0) node [left]  {\color{blue}$1_4$}
(7.7,0.7) node [above right] {\color{blue}$0$}
(6.3,0.7) node [above left] {\color{blue}$0^3$}
(7.7,-0.7) node [below right] {\color{blue}$0^3$}
(6.3,-0.7) node [below left] {$ $};

\draw (13,0) circle (1);
\draw (14,0) node [right]  {$1_2$ \color{blue}($1_4$)}
(13,1) node [above] {$1_1$ \color{blue}($1_3$)}
(13,-1) node [below] {$1_3$ \color{blue}($1_1$)}
(12,0) node [left]  {$1_4$ \color{blue}($1_2$)}
(13.7,0.7) node [above right] {$0^3$ \color{blue}()}
(12.3,0.7) node [above left] {$0^2$ \color{blue}($0^3$)}
(13.7,-0.7) node [below right] {\color{blue}($0^3$)}
(12.3,-0.7) node [below left] {$0^2$ \color{blue}($0$)};

\end{tikzpicture}\label{fig:overlay}
\caption{Example for Definition~\ref{def:2}: $u=1_10^31_21_30^21_40^2$ (left),  $v=1_101_20^31_31_40^3$ (center), 1-overlay of $u$ on $v$ for $i=2$ (right).}
\end{figure}
\end{center}

In other words, a 1-overlay can be seen as an order-preserving bijection between 1's in the words $u$ and $v$, or simply as a shift of indices of 1's in $v$ relative to in $u$. This bijection between 1's induces a bijection between blocks of 0's in the following sense. If 
$$u = 1_10^{\alpha_{u,1}}1_20^{\alpha_{u,2}}\cdots 1_{n_1}0^{\alpha_{u,n_1}}, \quad v = 1_10^{\alpha_{v,1}}1_20^{\alpha_{v,2}}\cdots 1_{n_1}0^{\alpha_{v,n_1}},$$
with $\alpha_{u,i},\alpha_{v,i}\geqslant 0$, then $0^{\alpha_{u,i}}$ is translated to  $0^{\alpha_{v,i+j}}$. We remark that $\alpha_{u,i}$ and $\alpha_{v,i+j}$ can be equal to 0. If $\alpha = 0$, we say that $0^{\alpha}$ is {\it empty place} or $\varnothing$ .

In addition, if $u'$ and $v'$ are equal subwords of $u$ and $v$, respectively, and $n_{1,u'} = n_{1,u} = n_{1,v} = n_{1,v'}$, then we can consider a 1-overlay of $u$ on $v$ such that $u'$ and $v'$ coincide. In other words, if
$$u' = 1_10^{\beta_{u,1}}1_20^{\beta_{u,2}}\cdots 1_{n_1}0^{\beta_{u,n_1}}, \quad v' = 1_10^{\beta_{v,1}}1_20^{\beta_{v,2}}\cdots 1_{n_1}0^{\beta_{v,n_1}};$$
where $0 \leqslant  \beta_{u,j} \leqslant  \alpha_{u,j}$, $0 \leqslant  \beta_{v,j} \leqslant  \alpha_{v,j}$ and $\beta_{u,j} = \beta_{v,j+i}$, then a 1-overlay of $u$ on $v$ can be considered as a bijection between indexed 1's in $u$ and in $v$ which translates $1_j$ in $u$ to $1_{j+i}$ in $v$. For example, if  $u = 1_10^21_201_30$, $u' = 1_10^21_201_3$, $v = 1_11_20^31_30$, and $v' = 1_20^21_301_1$, then for a 1-overlay of $u$ on $v$ with $i=1$ the subwords $u'$ and $v'$ coincide: $\frac{1_1\, 0^2 \, 1_2\, 0\, 1_3\, 0}{1_2 \,0^3 \, 1_3 \, 0\,  1_1\, \,\,}$.

\begin{definition} We say that a subword $s$ of a cyclic word $w$ is \emph{unioccurrent} if there is only one occurrence of $s$ in $w$, i.e., if we index 0's and 1's in $s$ and in $w$, then there is a unique injection from indices of letters in $s$ to indices of letters in $w$ giving an occurrence of $s$ in $w$.
\end{definition}

\begin{example}
Let $w = 01011 = 0_11_10_21_21_3$. Then the word $111$ is not unioccurrent, since we can choose it in three ways $1_11_21_3, 1_21_31_1, 1_31_11_2$, and the word $0110$ is unioccurrent, since we can take in one way $0_21_21_30_1$.
\end{example}

We remark that in a periodic cyclic word $w$ there are no unioccurrent subwords: indeed, we can shift the indices by the period. Similarly, a periodic subword of any cyclic word cannot be unioccurrent.

\begin{definition}\label{def:turn}
Let $w$ be a cyclic word with indexed 1's: $w = 1_1 0^{\alpha_1} 1_2 0^{\alpha_2}1_3 0^{\alpha_3}\cdots 1_{n_1} 0^{\alpha_{n_1}}$ ($\alpha_i \geqslant  0$). 
A {\emph{turn}} of $w$ is a word of the form $1_1 0^{\alpha_i} 1_2 0^{\alpha_{i+1}}1_3 0^{\alpha_{i+2}}\cdots 1_{n_1} 0^{\alpha_{i-1}}$ for some $i$, where indices are taken modulo $n_1$ (from 1 to $n_1$). 
\end{definition} 


An example illustrating the above definition is provided on Fig.~\ref{fig:turn}.

\begin{center}
\begin{figure}\label{fig:turn}
\begin{tikzpicture}

\draw (1,0) circle (1);
\draw (2,0) node [right]  {$1_2$}
(1,1) node [above] {$1_1$}
(1,-1) node [below] {$1_3$}
(0,0) node [left]  {$1_4$}
(1.7,0.7) node [above right] {$0^3$}
(0.3,0.7) node [above left] {$0$}
(1.7,-0.7) node [below right] {$0$}
(0.3,-0.7) node [below left] {$0^2$};

\draw (5,0) circle (1);
\draw (6,0) node [right]  {$1_2$}
(5,1) node [above] {$1_1$}
(5,-1) node [below] {$1_3$}
(4,0) node [left]  {$1_4$}
(5.7,0.7) node [above right] {$0$}
(4.3,0.7) node [above left] {$0^2$}
(5.7,-0.7) node [below right] {$0^3$}
(4.3,-0.7) node [below left] {$0$};

\draw (9,0) circle (1);
\draw (10,0) node [right]  {$1_2$}
(9,1) node [above] {$1_1$}
(9,-1) node [below] {$1_3$}
(8,0) node [left]  {$1_4$}
(9.7,0.7) node [above right] {$0^2$}
(8.3,0.7) node [above left] {$0$}
(9.7,-0.7) node [below right] {$0$}
(8.3,-0.7) node [below left] {$0^3$};

\draw (13,0) circle (1);
\draw (14,0) node [right]  {$1_2$}
(13,1) node [above] {$1_1$}
(13,-1) node [below] {$1_3$}
(12,0) node [left]  {$1_4$}
(13.7,0.7) node [above right] {$0$}
(12.3,0.7) node [above left] {$0^3$}
(13.7,-0.7) node [below right] {$0^2$}
(12.3,-0.7) node [below left] {$0$};

\end{tikzpicture}
\caption{Example for Definition~\ref{def:turn}: four turns of the word $1_10^31_201_30^21_40$.}
\end{figure}
\end{center}

Essentially, a turn of $w$ corresponds to a 1-overlay from Definition~\ref{def:2} of a word with itself. Note that a turn of the word $w$ generates a turn of its subwords in the following sense. Let $w$ be a cyclic word with indexed 1's: $w = 1_1 0^{\alpha_1} 1_2 0^{\alpha_2}1_3 0^{\alpha_3}\cdots 1_{n_1} 0^{\alpha_{n_1}}$ ($\alpha_i \geqslant  0$). Let $s_1 = 1_1 0^{\beta_1} 1_2 0^{\beta_2}1_3 0^{\beta_3}\cdots 1_{n_1} 0^{\beta_{n_1}}$, $\beta_i \leqslant  \alpha_i$, and $s_2 = 1_1 0^{\gamma_1} 1_2 0^{\gamma_2}1_3 0^{\gamma_3}\cdots 1_{n_1} 0^{\gamma_{n_1}}$, $\gamma_i \leqslant  \alpha_i$ be two subwords of $w$ containing all 1's. We say that $s_1$ and $s_2$ \textit{differ by a turn} if there exists a turn of $w$ such that it generates a turn of $s_1$ which is equal to $s_2$, i.e., there exists $j$ such that $\beta_i = \gamma_{i+j}$ for each $i$ ($i+j$ is taken modulo $n_1$). We further say that the turn of $w$ translates the subword $s_1$ to the equal subword $s_2$. We would like to emphasize that for the notion of a turn, letters are not treated symmetrically.

In the proof of the main result, we make use of the following notation:

\begin{definition}\label{def_cong} Let $u$ and $v$ be cyclic words of equal length, such that $n_{1,u} = n_{1,v}$ and $x_u=x_v = x$.  Let $u_1$ and $v_1$ be subwords of $u_{short}$ and $v_{short}$ correspondingly. We write $u_1 \cong v_1$ if $u_1 = v_1$ and we can add one block $0^x$ to $u_1$ and one block $0^x$ to $v_1$ such that the obtained words are subwords of $u$ and $v$, respectively, and they are equal. \end{definition} 
    
Note that if there are no such blocks $0^x$, then $u_{\short} \ncong v_{\short}$, even if $u_{\short} = v_{\short}$.

\begin{example}
    Let $x > 1$ and $u = 0^x 1011 0^x 1011 111$, $v = 0^x 1011 1011 0^x 111$ and $w = 0^x 1101 0^x 1101 111$. Then $u_{\short} = v_{\short} = w_{\short} = 1011 1011 111$ and $u_{\short} \cong v_{\short}$, since we can add a block $0^x$ to $u_{\short}$ and add a block $0^x$ to $v_{\short}$ and get the word $0^x 1011 1011 111$. However, we have $u_{\short} \ncong w_{\short}$ and $v_{\short} \ncong w_{\short}$.
\end{example}

\begin{remark}
    This notion is used only for special words.
\end{remark}

\section{Preliminary observations and auxiliary statements}

In the following proposition, we show that for each $n \geqslant 7$ there are pairs of words for which the sets of their subwords of length at most $\frac34n-\frac32$ are equal.

\begin{proposition}\label{pr:lower_bound}
    For $n = 4m+3$, $n=4m+4$, $n=4m+5$ and $n = 4m+6$  the following pairs of words:
    $$0^{m-2}1 0^{m-1} 1 0^{m+1}10^{m+1}1, \quad 0^{m-1}1 0^{m-2}1 0^{m+1}10^{m+1}1;$$ 
    $$0^{m-2}1 0^m 1 0^{m+1}10^{m+1}1, \quad 0^{m}1 0^{m-2}1 0^{m+1}10^{m+1}1;$$ 
    $$0^{m-1}1 0^m 1 0^{m+1}10^{m+1}1, \quad 0^{m}1 0^{m-1}1 0^{m+1}10^{m+1}1;$$ 
    $$0^{m-2}1 0^m1 0^{m+2}10^{m+2}1, \quad 0^{m}1 0^{m-2}1 0^{m+2}10^{m+2}1;$$ 
    can be distinguished by subwords of length $3m+2$, $3m+3$, $3m+5$ and $3m+3$, respectively, and not smaller than that.

\bigskip
\end{proposition}

\begin{proof} 

We provide a proof for the case $n = 4m+5$; the proofs for the other cases are similar.

Consider the words $u = 0^{m-1}1 0^m 1 0^{m+1}10^{m+1}1$ and $v = 0^{m}1 0^{m-1}1 0^{m+1}10^{m+1}1$. Let $w$ be a distinguishing subword for $u$ and $v$. 
We need to prove that $|w|\geqslant 3m+5$. Without loss of generality we can assume that  $w$ is a subword of the word $u$ and is not a subword of $v$ (the case when $w$ is a subword of $v$ and not a subword of $u$ is symmetric, since $u$ and $v$ are mirror images of each other as cyclic words).
Since the word $0^{m-1}1 0^{m-1} 1 0^{m+1}10^{m+1}1$ is a subword of $v$, $w$ must include the second block $0^m$ of $u$. Indeed, if $w$ does not include the second block $0^m$, then $w$ is a subword of $0^{m-1}1 0^{m-1} 1 0^{m+1}10^{m+1}1$, and hence it is a subword of $v$. Similarly, since the words $0^{m-1}1 0^m 1 0^{m+1}10^{m}1$ and $0^{m-1}1 0^m 1 0^{m-1}10^{m+1}1$ are subwords of $v$, then $w$ must contain the fourth block $0^{m+1}$ of $u$ 
and must contain at least $m$ zeros from the third block $0^{m+1}$ of $u$. Besides that, $w$ must have at least three 1's, because each subword of $u$ which has at most two occurrences of 1 is a subword of $v$. So, $|w| \geqslant (m+(m+1)+m+3)+1 = 3m+5$. To finish the proof, it remains to notice that the word $10^{2m+1}10^{m+1}1$ of length $3m+5$ is a subword of the word $u$ and is not a subword of $v$.

For other cases we provide words which are subwords of $u$ and are not subwords of $v$. For $n = 4m+3$ one can take the subword  $10^{m-1}10^{m-1}10^{m}1$ of length $3m+2$, for $n = 4m+4$ the subword $10^{m-1}10^{m-1}10^{m+1}1$  of length $3m+3$, $n = 4m+6$ the subword $10^{m-1}10^{m-1}10^{m+1}1$ of length $3m+3$. \qed

\end{proof}

\begin{proposition}\label{prop n_0, n_1, l}
Let $u$ be a cyclic word of length $n$ and $S$ be the set of its subwords of length at most $\frac34n+4$. Then we can recover $n_0$, $n_1$ and $l$ from $S$ and $n$.
\end{proposition}

\begin{proof} Without loss of generality assume that the number of 1's in $u$ is less than or equal to the number of 0's (we can check which letter is more frequent by checking if the word $0^{\lfloor \frac{n}{2}\rfloor+1}$ is a subword). Then $S$ contains the word $1^{n_1}$ and does not contain the word $1^{n_1+1}$. So we can recover $n_1$, and $n_0 = n - n_1$.

It is clear that $l \leqslant  n_1 \leqslant  \frac{n}2$.
If $l \leqslant  \frac38 n,$ then $S$ contains the word $(01)^{l}$. In this case we can find the word $(01)^k$ in $S$ with maximal $k$ and recover $l$. 

If $l > \frac38 n$, then we can consider the subword $u' = 0^{a_1}1^{b_1} \cdots 0^{a_k}1^{b_k}$ of the word $u$, which has $n_0$ occurrences of 0 and each block of 1's containing at least two occurrences of 1. Let $l_{1}$ be the number of blocks 1 which have length 1. Then $l_1 + k = l \geqslant  \frac38 n$ and $l_1 + 2k \leqslant  n_1 \leqslant  \frac{n}{2}$. So, $\frac{3}{4}n \leqslant  2l_1 + 2k \leqslant  l_1 + n_1 \leqslant  l_1 + \frac{n}{2}$. It means that $l_1 \geqslant  \frac{n}{4}$, and hence $|u'| = n - l_1 \leqslant  \frac34n$. Thus, $u'$ is in $S$. Moreover, we can detect it in $S$ as a word of the form $0^{a_1}1^{b_1} \cdots 0^{a_k}1^{b_k}$, with $n_0$ occurrences of $0$, $b_i \geqslant  2$ for all $1 \leqslant  i \leqslant  k$, and with $k$ and $b_i$ being maximal.
So, we can recover $l$ from $u'$: $l = k+l_1 = k + (n-|u'|)$. \qed
\end{proof}

\begin{corollary} \label{cor:parameters}
Let $u$ and $v$ be two cyclic words of length $n$. If $n_{0,u} \neq n_{0,v}$ or $n_{1,u} \neq n_{1,v}$ or $l_{u} \neq l_{v}$, 
then there exists a distinguishing subword of length at most $\frac34n+4$ for $u$ and $v$.
\end{corollary}

\proof Follows from Proposition~\ref{prop n_0, n_1, l}. \qed

\bigskip

The following proposition and corollary are key tools for the proof of the main theorem for the case of neither special nor periodic words (we use it in Lemma~\ref{lemma 2}).

\begin{proposition}
Let $u = 01 ^{\alpha_1}01^{\alpha_2}\cdots01^{\alpha_l}$ be a cyclic word which is neither special nor periodic, with $l \geqslant  2$, and $\alpha_i \geqslant  1$ for each $1 \leqslant  i \leqslant  l$. Then there exists a unioccurrent subword $u_{\fix}$ of $u$ which contains $n_{1,u} = n_1$ occurrences of 1 (all 1's from $u$) and at most $\frac{l+2}{2}$ occurrences of 0.
\end{proposition}

\begin{proof}

Let $k$ be the minimal length of blocks of 1's: $k=\min_{1\leqslant  i \leqslant  l} \alpha_i$, and let $s$ denote the maximal integer such that $(01^k)^s0$ is a factor of $u$. So,  $u = (01^k)^s01^{\alpha_1}\cdots 01^{\alpha_r}$, where $s + r = l$, $\alpha_i \geqslant  k$ for each $i$. Since $s$ is chosen maximal, we have in particular that $\alpha_1, \alpha_r > k$. Since $u$ is not special, we have $s < l$. There are two cases to consider:

\medskip

\noindent {\bf Case 1.} $s \leqslant  \frac{l}2$. 

\medskip

Consider the subword $u' = (01^k)^s01^{n_1-ks}$ of $u$. 

Since $s$ is chosen maximal, for any two occurrences of subwords equal to $u'$, there is no occurrence of $0$ that is included in both subwords. The words $u$ and $u'$ have $l$ and $s+1$ occurrences of 0, respectively. So, there are at most  $\frac {l}{s+1}$ ways to take an occurrence of a subword of $u$ which is equal to $u'$.

Consider an occurrence of a subword $w$ in $u$ which is equal to $u'$. Since $u$ is not periodic, for every turn of $u$ which translates $u'$ to $w$ we can find a zero in $u$ which is translated to an empty place. Consider a word $w'$ obtained by adding all these 0's to $w$. Since there are at most  $\frac {l}{s+1}$ ways to take an occurrence of a subword of $u$ which is equal to $u'$, $w'$ has $n_1$ occurrences of 1 and at most $s+1 +(\frac {l}{s+1}-1) = s+\frac{l}{s+1}$ occurrences of 0, and $w'$ is unioccurrent. For $1 \leqslant  s \leqslant  \frac{l-2}{2}$ we have $w'_0 \leqslant  s+\frac{l}{s+1} \leqslant  \frac{l+2}2$. For $\frac{l-2}{2} < s \leqslant  \frac{l}{2}$, the subword $u'$ is unioccurrent and $u'_0 = s+1 \leqslant  \frac{l+2}2$. In both cases we find a unioccurrent subword with $n_1$ occurrences of 1 and at most $\frac{l+2}2$ occurrences of 0.

\medskip

\noindent {\bf Case 2.} $s > \frac{l}2$.

\medskip

We spit this case into three subcases as follows. Case {\bf 2.1} corresponds to $k \nmid \alpha_i$ for some $i$. If $k \mid \alpha_i$, we distinguish between two cases: either $u = (01^k)^s01^{k\beta_1}01^{k\beta_2}\cdots 01^{k\beta_r}$ and $\beta_1 \geqslant  3$ (or $\beta_r \geqslant  3)$ (Case {\bf{2.2}}), or $u = (01^k)^s01^{k\beta_1}01^{k\beta_2}\cdots 01^{k\beta_r}$, $\beta_1 = \beta_r = 2$ (Case {\bf 2.3}); in the latter case we have $\beta_i \geqslant  2$ for some $1 < i < r$ since $u$ is not special. We now consider these three subcases.

\medskip

\noindent {\bf Case 2.1.} $k \nmid \alpha_i$.

\medskip

Let $1^t$ be the smallest block of 1's such that $k \nmid t$. Then $t$ cannot be equal to the sum of two or more integers from $\{k, \alpha_1, \ldots, \alpha_r\}$. Let $s'$ be the maximal number of consecutive blocks $1^t$. Since $s > \frac{l}2$, we have $s' < \frac{l}2$. Consider the word $u' = (01^t)^{s'}01^{n_1-ts'}$. There are at most $\frac{l}{s'+1}$ ways to take an occurrence of a subword of $u$ which is equal to $u'$. So, with an argument similar to Case {\bf 1} we prove that we can find a unioccurrent subword with $n_1$ occurrences of 1 and at most $\frac{l+2}2$ occurrences of 0.

So, it remains to consider subcases with  $\alpha_i = k\beta_i$.

\medskip

\noindent {\bf Case 2.2.} $u = (01^k)^s01^{k\beta_1}01^{k\beta_2}\cdots 01^{k\beta_r}$ and $\beta_1 \geqslant  3$ (or $\beta_r \geqslant  3)$ 

\medskip

Without loss of generality we assume that $\beta_1\geqslant  3$. Recall that in Case {\bf 2} we have $s>r$, since $s>l/2$. Consider the word $$u'= \begin{cases}(01^k)^r(1^k01^k)^{\frac{s-r}{2}}01^{k\beta_1}1^{k\beta_2}\cdots 1^{k\beta_r}, & \text{if } 2 \mid (s-r), \\  (01^k)^r(1^k01^k)^{\frac{s-r-1}{2}}1^k01^{ k\beta_1}1^{k\beta_2}\cdots 1^{k\beta_r}, & \text{if } 2 \nmid (s-r).\end{cases}$$ 
In both cases $u'$ has $n_1$ occurrences of 1 and at most $r+\frac{s-r}2+1 = \frac{s+r+2}{2} =  \frac{l+2}2$ occurrences of 0. 
For example, if $u = \underline{0}1\underline{0}1\underline{0}1\,01\underline{0}1\,01\underline{0}1\,\underline{0} \, 1^30101^2$ ($ r = 3$, $s = 7$, $k = 1$), then $u'$ contains all $1$'s and underlined $0$'s: $u' = \underline{0}1\underline{0}1\underline{0}1\,1\underline{0}1\,1\underline{0}1\,\underline{0} \, 1^311^2$.
As another example, take $u = \underline{0}1\underline{0}1\underline{0}1\,01\underline{0}1\,01\underline{0}1\,01\underline{0} \, 1^30101^2$ ($ r = 3$, $s = 8$, $k = 1$); then $u'$ contains all $1$'s and underlined $0$'s: $u' = \underline{0}1\underline{0}1\underline{0}1\,1\underline{0}1\,1\underline{0}1\,1\underline{0} \, 1^311^2$.

Now we prove that $u'$ is unioccurrent in the case $2 \mid (s-r)$ (the case $2 \nmid (s-r)$ is similar). Assume the converse: suppose that there exists 
another occurrence $u''$ of the subword $u'$, i.e. $u''=u'$ and a turn $\sigma$ of $u$ which translates $u'$ to $u''$. We now index 0's in $u$: $u = 0_11^k0_21^k\cdots 0_s1^k0_{s+1}1^{k\beta_1}0_{s+2}1^{k\beta_2}\cdots 0_{s+r}1^{k\beta_r}$. If $\sigma(0_1) = 0_2$, then $\sigma(0_{s+1}) = \varnothing$, a contradiction. If $\sigma (0_1) = 0_3$ or $\sigma(0_1) = 0_4$, then $\sigma (0_{s}) = \varnothing$, a contradiction. If $\sigma (0_1) = 0_5$ or $\sigma(0_1) = 0_6$, then $\sigma (0_{s-2}) = \varnothing$, a contradiction. Continuing this line of reasoning, we get that  $\sigma(0_1) \neq 0_2,0_3, \ldots, 0_{s+1}$. Similarly, $\sigma(0_1) \neq 0_{s+2}, 0_{s+3}, \ldots, 0_{s+r}$, since in this case there  exists $1 \leqslant  i \leqslant  r-1$, such that $\sigma (0_i) = 0_{s+r}$ and $\sigma(0_{i+1}) = 0_1$. Recall that the number of 1's between $0_i$ and $0_{i+1}$ is equal to $k$, and $\beta_r > k$ is the number of 1's between $0_{s+r}$ and $0_1$. We reach a contradiction, since both $u'$ and $u''$ have $n_1$ occurrences of 1 and $0_1, 0_2, \ldots,0_r$. So, $\sigma$ is the identical turn, and hence $u'$ is a unioccurrent subword. 

\medskip
\noindent {\bf Case 2.3.} $u = (01^k)^s01^{k\beta_1}01^{k\beta_2}\cdots 01^{k\beta_r}$, $\beta_1 = \beta_r = 2$ and $\beta_i \geqslant  2$ for some $1 < i < r$.

\medskip

So, $u = (01^k)^s01^{2k}01^{k\beta_2}\cdots 01^{2k}$. We can rewrite $u$ in the form $$u = (01^k)^s0 1^{2k}(01^k)^{b-1}01^{k\beta_{b+1}}\cdots 01^{k\beta_j}(01^k)^{a-1}01^{k\beta_j+a+1}\cdots01^{2k},$$ where $a,b \geqslant  1$, $\beta_{b+1}, \beta_j, \beta_{j+a+1} > 1$, and $a$ is chosen maximal. So, there are $b-1$ consecutive blocks $1^{k}$ after the first block $1^{2k}$, and $a-1$  is the maximal number of consecutive blocks $1^{k}$ in the remaining part. Without loss of generality we can assume that $a \geqslant  b$ (indeed,  otherwise we can take a mirror image of $u$: $(01^k)^s01 ^{k\beta_r}01^{k\beta_{r-1}}\cdots 01^{k\beta_1}$). We can also suppose that the parts $(01^k)^{b-1}$ and $(01^k)^{a-1}$ do not coincide, since there exists $\beta_i \geqslant  2$. Let $t$ be an integer such that $s = (a+1) + t(2b+2) + x$, where $x < 2b+2$. 

Consider the following subword of $u$:
$$u' = \begin{cases} (01^k)^{a+1}((1^k)^{b+1}(01^k)^{b+1})^{t}(01^k)^{x}01^{k\beta_1}1^{k\beta_2}\cdots1^{k\beta_r}, & \text{ if } x \leqslant  b, \\ (01^k)^{a+1}((1^k)^{b+1}(01^k)^{b+1})^{t}(1^k)^{x-b}(01^k)^{b}01^{k\beta_1}1^{k\beta_2}\cdots1^{k\beta_r}, & \text{ if } b < x < 2b+2.\end{cases}$$

For example, if $u = 010101010101010 \, 1^201^20101^2$ ($ r = 4$, $a = 2$, $b = 1$, $k = 1$, $s = 7 = 3 + 1\cdot 4+0$), then $u'$ contains all 1's from $u$ and underlined 0's: $u = \underline{0}1\underline{0}1\underline{0}1\,0101\underline{0}1\underline{0}1\,\underline{0} \, 1^201^20101^2$. 

The word $u'$ has $n_1$ occurrences of 1 and $a+1 + t(b+1) + \min(b,x) + 1 \leqslant  \frac{l+2}2$ occurrences of 0, since $$l \geqslant  s+1+a+b =  a+1 + t(2b+2) + x + 1 + (a + b) =$$ $$= 2(a+1) + 2t(b+1) + (b+x) \geqslant  2(a+1 + t(b+1) + \min(b,x) + 1) - 2.$$

We now prove that $u'$ is unioccurrent in the case $2 \mid (s-r)$ (in the case $2 \nmid (s-r)$ the proof is similar). Assume the converse, i.e. suppose that there exists a subword $u''$ such that $u'' = u'$ and a turn of $u$ which translates $u'$ to $u''$. We index 0's and 1's in $u$ as follows: $$u = 0_11^k0_21^k\cdots 0_s1^k0_{s+1}1^{2k}0_{s+2}1^{k\beta_2}\cdots 0_{s+r}1^{2k}.$$

If $\sigma(0_1) = 0_{s+1}$, then $\sigma(0_{2}) = \varnothing$, a contradiction. 

If $\sigma(0_1) \in  \{0_{s}, 0_{s-1}, \ldots, 0_{s-(a-2)}\}$, then for $\sigma(0_1) = 0_{s-i}$ we have $\sigma(0_{i+3}) = \varnothing$ and $0_{i+3} \in u'$. A contradiction.

If $\sigma(0_1) = 0_{s-(a-2)-t'(2b+2)-x'}$ and $s-(a-2)-t'(2b+2)-x' > 0$ for $1 \leqslant  x' \leqslant  b+1$, then $\sigma(0_{a+1+t'(2b+2)+(b+1)+x'} )= \varnothing$ and $0_{a+1+t'(2b+2)+(b+1)+x'} \in u'$. A contradiction.

If $\sigma(0_1) = 0_{s-(a-2)-t'(2b+2)-x'}$ and $s-(a-2)-t'(2b+2)-x'$ for $b+2 \leqslant  x' \leqslant  2(b+1)$, then $\sigma(0_{a+1+t'(2b+2)+x'}) = \varnothing$ and $0_{a+1+t'(2b+2)+x'} \in u'$. A contradiction.

We also have $\sigma(0_1) \not\in \{0_{s+2},0_{s+3}, \ldots, 0_{s+r}\}$, since otherwise the turn of the subword $0_1 1^k 0_2 1^k \cdots 0_{a}1^k0_{a+1}$ is translated to another subword of $\sigma(0_1 1^k 0_2 1^k \cdots 0_{a}1^k0_{a+1}) = 1^{2k}0_{s+2}1^{k\beta_2}\cdots 0_{s+r}1^{2k}$, a contradiction with the maximality of $a$.

So, $u'$ is unioccurrent. \qed

\end{proof}

\begin{corollary} \label{main corollary}
Let $u = 0^x1 ^{\alpha_1}0^x1^{\alpha_2}\cdots0^x1^{\alpha_l}$ be a cyclic word which is neither special nor periodic, with $l \geqslant  2$, and $\alpha_i \geqslant  1$ for each $1 \leqslant  i \leqslant  l$. Then there exists a unioccurrent subword $u_{\fix}$ of $u$ which contains $n_1$ occurrences of 1 (all 1's from $u$) and at most $\frac{l+2}{2}$ blocks $0^x$ (and no other zeros).  
\end{corollary}

\begin{proof}
The previous proposition implies that the word $01^{\alpha_1}01^{\alpha_2}\cdots01^{\alpha_l}$ has a unioccurrent subword $01^{\beta_1}01^{\beta_2}\ldots01^{\beta_r}$, which has $n_1$ occurrences of 1 and at most $\frac{l+2}{2}$ occurrences of 0. Then the word $0^x1^{\beta_1}0^x1^{\beta_2}\cdots0^x1^{\beta_r}$ is a unioccurrent subword of $u$. \qed
\end{proof}

\bigskip

In the proof of the main result, for finding a distinguishing subword for the words $u$ and $v$, we often use a technique described in the following proposition:

\begin{proposition}  \label{prop popular idea}
    Let $u \neq v$ be two cyclic words such that $n_{0,u} = n_{0,v}$, $n_{1,u} = n_{1,v} = n_1$, ${x_u} = {x_v} = x$. Suppose that $u'$ is either a unioccurrent subword of $u$ or a subword of $u$ which contains $n_1$ occurrences of 1 and such that the only turn translating $u'$ to an occurrence of an equal subword is the identity map. Suppose also that $|u'|+y+1\leqslant  \frac34n+4$, where $y$ is the length of the longest block in $u$ which is shorter than $0^x$. Then there exists a distinguishing subword of length at most $\frac34n+4$ for $u$ and $v$.
\end{proposition} 

\begin{remark} 
     Note that in the second case, when $u'$ is a subword of $u$ which contains $n_1$ occurrences of 1 and such that the only turn translating $u'$ to an equal subword is the identity map, $u'$ does not have to be unioccurent, since we can choose different $0$'s from blocks of length greater than $1$.
\end{remark}

\begin{proof}
    Note that either $v$ does not have a subword which is equal to $u'$ or there is a subword $v' = u'$ of $v$. In the first case $u'$ is a subword of the word $u$ and is not a subword of $v$ and $|u'| \leqslant  \frac34n+3$. In the second case we can consider a 1-overlay of $v$ on $u$ such that $v'$ and $u'$ coincide. Since $v \neq u$, there is a block $0^s$ in $u$ ($s$ is the length of the block, and we set $s=0$ if the block is empty) which is smaller than the corresponding block of 0's in $v$ (we let $y$ denote the length of this block, so that $s \leqslant  y$). We add the block $0^{s+1}$ to $v'$ to the corresponding place. We get a subword $v''$  of the word $v$ which is not a subword of $u$, and $|v''| \leqslant  |v'| + s+1  \leqslant  |v'|+y+1 \leqslant  \frac34n+4$. \qed
\end{proof}

\section{Proof of Theorem~\ref{th:main}.}

In this section, we provide a proof of Theorem~\ref{th:main}. In Subsection~\ref{subsec:proof_structure} we give a general structure of the proof, splitting it into several lemmas, and introduce some auxiliary notation used throughout the proof.  In Subsection~\ref{subsec:proof_lemmas}, we state and prove lemmas constituting the proof.

\subsection{Notation and general structure of the proof} \label{subsec:proof_structure}

In this subsection, we fix some notation and give a general description of the proof of the main result of this paper, Theorem~\ref{th:main}. Namely, let $u \neq v$ be two cyclic words of length $n$. We will prove that there is a word $w$ of length at most $\frac34n+4$ such that $w$ is a distinguishing subword for the words $u$ and $v$. By Corollary~\ref{cor:parameters}, it remains to prove the theorem when $n_{0,u} = n_{0,v} = n_0$, $n_{1,u} = n_{1,v} = n_1$ and $l_u = l_v = l$. Without loss of generality we can assume that $n_0 \geqslant  n_1$ (or, equivalently, $n_1 \leqslant  \frac{n}2$).

Recall that $0^{x_u}$ and $0^{x_v}$ are the longest blocks of 0's in $u$ and $v$, respectively, and that $y = y_u$ is the length of the longest block of 0's in $u$ which is shorter than $0^{x_u}$ (if all the blocks have the same length, we set $y = 0$). We let $a$ and $b$ denote the numbers of blocks $0^{x_u}$ and $0^{x_v}$ in $u$ and $v$, respectively. We  divide the proof of our theorem into five lemmas according to different cases as follows:

\begin{itemize}

\item $x_u \neq x_v$ (Lemma~\ref{lemma 1}),

\item $u_{\longg}$ is a neither special nor periodic word with $a \geqslant  3$, or, analogously, $v_{\longg}$ is a neither special nor periodic word with $b \geqslant  3$
(Lemma~\ref{lemma 2}).

\item $u_{\longg}$ is periodic and not special, $a \geqslant  3$, or, analogously, $v_{\longg}$ is periodic and not special, $b \geqslant  3$ (Lemma~\ref{lemma 3}).

\item one of the words $u_{\longg}, v_{\longg}$ is not a special word with at most two blocks $0^x$, and the other word is either special or contains at most two blocks $0^x$ (Lemma~\ref{lemma 4}).

\item
 $u_{\longg}$ and $v_{\longg}$ are special words (Lemma~\ref{lemma 5}).
   
\end{itemize}

It is not hard to see that all the cases are covered. In the first lemma we prove the theorem for words which have different sizes of big blocks. In the second and third lemmas we prove the theorem in the case when one of the words is not special and has at least three big blocks. In the fourth and fifth lemmas we prove the theorem in the case when either both words are special or one of the word is not special and has one or two big blocks. 

\subsection{Lemmas constituting the proof} \label{subsec:proof_lemmas}

In this subsection, we prove five lemmas corresponding to different cases of the proof of Theorem~\ref{th:main}.

\begin{lemma} \label{lemma 1}
    Let $x_u \neq x_v$. Then $u$ and $v$ have a distinguishing subword of length at most $\frac34n+4$.
\end{lemma} 

\noindent\textit{Proof.} We consider two cases: 

\medskip

\noindent {\bf Case 1.} $a \geqslant  2$ or $b \geqslant  2$.

\medskip

Without loss of generality assume that $a \geqslant  2$. If $x_u > x_v$, then consider the word $u_1 = 1^{n_1}0^{x_u}$. This word is the subword of the word $u$ and is not a subword of $u$. Since $a \geqslant  2$,
we have $x_u \leqslant  \frac{ax_u}2 \leqslant  \frac{n_0}2$. So, $|u_1| = n_1 + x_u \leqslant  n_1+\frac{n_0}{2} = \frac{n}2 + \frac{n_1}2 \leqslant  \frac34 n$. If $x_u < x_v$, then consider the word $u_1 = 1^{n_1}0^{x_u+1}$. This word is a subword of the word $v$ and is not a subword of $u$, and its length is $|u_1| = n_1 + x_u+1 \leqslant  n_1+\frac{n_0}{2} + 1 \leqslant  \frac34 n + 1$. In both cases we have a desired distinguishing subword.

\medskip

\noindent {\bf Case 2.} $a = b = 1$. 

\medskip

Without loss of generality we assume that $x_u > x_v$. The words $u_1 = 1^{n_1}0^{x_v+1}$ and $u_2 = (01)^{l-1}0^{x_v+1}1$ are subwords of the word $u$ and are not subwords of $v$. If $|u_1| \leqslant  \frac34n+4$ or $|u_2| \leqslant  \frac34n+4$, then we have a required subword. Otherwise $|u_1| = n_1+x_v+1 > \frac34n+4$ and $|u_2| = 2l+x_v > \frac34n+4$. If $x_v \leqslant  \frac{n_0}{2}$, then $|u_1| \leqslant  n_1+\frac{n_0}{2} + 1 \leqslant  \frac34 n + 1$. So, it remains to consider the case $x_u > x_v > \frac{n_0}2$ and each block of 0's in $u$ except for $0^{x_u}$ contains less than $x_v$ occurrences of 0 (in particular, the block $0^y$).

We now consider a 1-overlay of $v$ on $u$ such that the blocks $0^{x_u}$ and $0^{x_v}$ coincide. Since $x_u > x_v$, there are two neighboring 1's such that $v$ has more 0's between them than $u$ has for this overlay. We let $t$ and $s > t$  denote the numbers of occurrences of 0 between them in  $u$ and $v$, respectively. We let $p$ denote the number of occurrences of 1 between this block of 0's and the block $0^{x_v}$ in $v$. Consider the word $v_1=1^p0^{y+1}1^{n_1-p}0^{t+1}$. The word $v_1$ is a subword of the word $v$ and is not a subword of $u$. Since $u$ has at least $l+(x_u-1)+(y-1)$ occurrences of 0 (that is, $n_0 \geqslant l+(x_u-1)+(y-1)$) and $y \geqslant  t$, we have $$|u_1| + |u_2| + |v_1|=(n_1+x_v+1) + (2l+x_v) + (n_1+y+1+t+1) \leqslant $$ $$\leqslant  2n_1+2(l+x_v+y)+3 \leqslant  2n_1+2(n_0+2)+3 = 2n+7.$$ So, the length of at least one of the words $u_1$, $u_2$, $v_1$ is at most $\frac{2n+7}{3} < \frac{3}{4}n+4$.  \qed

\bigskip

In the following text we assume that  $x_u = x_v = x$. Note that since $n_{0,u} = n_{0,v} = n_0$, $n_{1,u} = n_{1,v} = n_1 \leqslant  \frac{n}2$ and $l_u = l_v$, we have $x \neq 1$. Otherwise  both words are equal to $(01)^{k}$. So, we also assume that $x \geqslant  2$. Recall that $u_{\longg}$ and $v_{\longg}$ are subwords of $u$ and $v$, respectively, which contain all 1's and all big blocks $0^x$.

\begin{lemma} \label{lemma 2}
Let $u_{\longg}$ be a neither special nor periodic word with $a \geqslant  3$. 
Then there exists a distinguishing subword of length at most $\frac34n+4$ for $u$ and $v$.
\end{lemma} 

\noindent\textit{Proof.} Due to Corollary~\ref{main corollary}, there exists a unioccurrent subword $u_{\fix}$ of $u_{\longg}$ such that $u_{\fix}$ has $n_1$ occurrences of 1 and at most $\frac{a}2 + 1$ blocks $0^x$. It is easy to see that $u_{\fix}$ is also a unioccurrent subword of $u$. The proof of the lemma is split into several cases as follows.  Cases {\bf 1} and {\bf 2} correspond to $a \geqslant 4$; Case {\bf 1} gives a proof under the condition $n_1+\frac{ax}2 + 2x < \frac34n+3$, and Case {\bf 2} treats the opposite inequality. Case {\bf 2} is divided into subcases {\bf 2.1}  and {\bf 2.2} corresponding to $a \geqslant  5$ and $a = 4$, respectively. Case {\bf 2.1} is further subdivided to subcases {\bf 2.1.1}  and {\bf 2.1.2} depending on whether  $u_{\zeros}$ is periodic or not. Case {\bf 3} corresponds to $a=3$.


















\medskip

We now proceed with the proofs in each of these cases.

\medskip

\noindent{\bf  Case 1.} {\it $n_1+\frac{ax}2 + x+y <  \frac34n+3$, $a \geqslant 4$.} 

\medskip

The subword $u_{\fix}$ of $u$ is unioccurrent and $|u_{\fix}|+y+1 \leqslant  n_1 + (\frac{a}2+1)x+y+1 \leqslant  n_1+\frac{ax}2 + x+y+1 <  \frac34n+4$. Then applying Proposition~\ref{prop popular idea} to the word $u_{\fix}$ we get that there is a distinguishing subword of length at most $\frac34n+4$ for $u$ and $v$.

\medskip
\noindent {\bf Case 2.} {\it $n_1+\frac{ax}2 + x+y \geqslant  \frac34n+3$, $a \geqslant 4$.} 

\medskip

\noindent {\bf Case 2.1.} $a \geqslant  5$. 

\medskip

Recall that $n_{\longg}$ and $n_{\short}$ are the numbers of 0's in big blocks (blocks $0^x$) and small blocks (blocks $0^{<x}$), respectively. Notice that $n_0 = n_{\longg} + n_{\short} = ax + n_{\short}$, $n = n_1 + n_0 = n_1 +ax + n_{\short}.$ So, 
\begin{equation} \label{eq:one} 
n_{\short} = n-n_1 - ax.
\end{equation} 
The word $u$ has $a$ blocks $0^x$ and at most $n_{\short} - y+1$ small blocks. So, 
\begin{equation} \label{eq:two} 
l \leqslant  n_{\short} - y + 1 + a.
\end{equation}
Since $n_1+\frac{ax}2 + x+y \geqslant  \frac34n+3$, \eqref{eq:one} and \eqref{eq:two} imply that 
$$l \leqslant  n_{\short} - y + 1 + a = n - n_1 - ax - y + a+1 = n - (n_1+\frac{ax}2+x+y) - \frac{ax}2+x+a+1 \leqslant$$ 
\begin{equation} \label{eq:three}
\\ \leqslant \frac{n}4 - \frac{ax-2x-2a+4}2 = \frac{n}4 - \frac{(a-2)(x-2)}{2}.
\end{equation}

Consider the words $u_{\zeros}$ and $v_{\zeros}$  (recall that they are subwords of $u$ and $v$ which contain all 1's and one zero from each block of 0's). Notice that $|u_{\zeros}| = n_1 + l \leqslant  \frac{n}2 + \frac{n}4 - \frac{(a-2)(x-2)}{2} \leqslant  \frac34n$, since $a \geqslant  2, x \geqslant  2$. If $v_{\zeros} \neq u_{\zeros}$, then $u_{\zeros}$ is a subword of the word $u$ and is not a subword of $v$, and $|u_{\zeros}| \leqslant \frac34n$. So it remains to prove the lemma in the case $v_{\zeros} = u_{\zeros}$. Consider two cases.

\medskip

\noindent {\bf Case 2.1.1.} $u_{\zeros}$ is not periodic. 

\medskip

In this case there is only one way to take an occurrence of $u_{\zeros}$ in $u$ modulo the selection of one zero from each block of 0's. Since 
$n_1 \leqslant  \frac{n}2$ and by \eqref{eq:three}, we have  $$|u_{\zeros}|+y+1 = n_1 + l + y+1 < n_1 + \frac{n}4 - \frac{ax-2x-2a+4}{2} + x + 1 \leqslant$$ 
\begin{equation} \label{eq:four}
\leqslant  \frac{n}2 + \frac{n}4 - \frac{ax-4x-2a+2}{2} =  \frac34n - \frac{(a-4)(x-2)}{2}+3 \leqslant  
\frac34n + 3,
\end{equation}
since $a \geqslant  5$ and $x \geqslant  2$. Applying Proposition~\ref{prop popular idea} to the word $u_{\zeros}$ we get that there is a distinguishing subword for the words $u$ and $v$ of length at most $\frac34n+4$.

\medskip
\noindent {\bf Case 2.1.2.} $u_{\zeros}$ is periodic. 

\medskip

The word $u_{\zeros}$ is of the following form: $u_{\zeros} = (1^{\alpha_1}01^{\alpha_2}0\cdots 1^{\alpha_t}0)^\frac{l}t$, where $\frac{l}t\geqslant 2$ is an integer and $t$ is minimal.
Consider the word $u_{\zeros}' = (1^{\alpha_1}01^{\alpha_2}0\cdots 1^{\alpha_t}0)(10)^{l-t}$. The word $1^{\alpha_1}01^{\alpha_2}0\cdots 1^{\alpha_t}0$ contains $\frac{n_1}{\frac{l}t}$ occurrences of 1, hence $u_{\zeros}'$ contains $\frac{n_1\cdot t}{l}+l-t$ occurrences of 1 and $l$ occurrences of 0. The function $f(t)=\frac{n_1\cdot t}{l}+l-t$ is increasing and $t \leqslant  \frac{l}2$ (since $\frac{l}t \geqslant  2$ is integer), so $u_{\zeros}'$ contains at most $\frac{n_1+l}{2}$  occurrences of 1. We add to $u_{\zeros}'$ all blocks $0^x$ from $u_{\fix}$ defined in the beginning of the proof. Let $u_{\zeros}''$ denote the obtained subword of $u$. There is only one way to take $u_{\zeros}''$ in $u$ modulo the selection of 1's and 0's from blocks from which we do not take all symbols. Since $l \leqslant  n_{\short}-y+a+1$, we have $n = n_1+ax+n_{\short} \geqslant  n_1+ax+l+y-a-1$. From these  inequalities and \eqref{eq:three} it follows that
$$|u_{\zeros}''|+y+1 \leqslant  \frac{n_1+l}{2} + l + \left(\frac{a+2}2\right)(x-1) + y +1=$$ $$= \frac{n_1+ax+l+y-a-1}{2} + l + x + \frac{y}2 +\frac{1}2< \frac{n}2 + \frac{n}4 - \frac{ax-2x-2a+4}{2} + \frac{3x}2=$$

\begin{equation}\label{eq:five}
    \frac34n - \frac{ax-2a-5x+4}{2} = \frac34n - \frac{(a-5)(x-2)-6}{2} \leqslant \frac34n+3
\end{equation} 
since $a \geqslant  5$ and $x \geqslant  2$.  Applying Proposition~\ref{prop popular idea} to the word $u_{\zeros}''$, we get that there exists a distinguishing subword of length at most $\frac34n+4$ for $u$ and $v$.

\medskip

\noindent {\bf Case 2.2.} $a = 4$.  

\medskip

In this case the proof is similar to the proof in Case {\bf 2.1}. All inequalities from Case {\bf 2.1} hold true, except for the inequality \eqref{eq:five} in Case {\bf 2.1.2}. If we prove that $l \leqslant  \frac{n}4 - \frac{ax-2x-2a+4}2 - \frac{y}2+1$, then 
we can rewrite inequality \eqref{eq:five} as

$$|u_{\zeros}''|+y+1\leqslant  \frac{n_1+ax+l+y-a-1}2 + l + x + \frac{y}2 + \frac12 < \frac{n}2 + \frac{n}4 - \frac{ax-2x-2a+4}2-$$ $$-\frac{y}2+1 + x + \frac{x}2=  \frac34n - \frac{ax-2a-4x+2}2 =  \frac34n - \frac{(a-4)(x-2)}{2}+3= \frac34n+3,$$ which gives us the proof in this case similarly to the proof in Case {\bf 2.1}.

It remains to prove that $l \leqslant  \frac{n}4 - \frac{ax-2x-2a+4}2 - \frac{y}2+1$. For $a = 4$ this is equivalent to the inequality \begin{equation}\tag{$3'$}l \leqslant  \frac{n}4-x+3-\frac{y}2. \end{equation} 

We now consider several cases according to the number of blocks of 0's of length at least $\frac{y}{2}$ in $u$ and in $v$. Since in Case {\bf 2.2}, the one we consider now, we have $a=4$, i.e., the number of blocks of 0's of length $x$ (the longest blocks) in $u$ is $4$, $u_{\longg}$ is not special by the conditions of the lemma, we have that the number of blocks of length at least $\frac{y}{2}$ in $u$ is at least $5$. 
 
If there are at least six blocks in $u$ with lengths at least $0^{\frac{y}{2}}$,  i.e. four blocks $0^x$, one block $0^y$ and at least one block $0^{\frac{y}2}$, then we can rewrite inequality \eqref{eq:three} as $l \leqslant  n_{\short} - y - (\frac{y}{2}-1)+ 1 + a$. So, in this case we have $l \leqslant  \frac{n}4 - \frac{ax-2x-2a+4}2 - \frac{y}2+1$. 

Now we assume that $u$ has only five blocks of 0's with lengths at least $\frac{y}2$ (four blocks $0^x$ and one $0^y$). The word $v$ can have  different numbers blocks of 0's with lengths at least $\frac{y}2$; we consider several cases accordingly.

If there are at least six blocks of 0's in $v$ which have lengths at least $\frac{y}2$, then we consider the subword of $v$ which contains $n_1$ occurrences of 1 and six blocks $0^{\frac{y}2}$. Then this word is a subword of the word $v$ and is not a subword of $u$, and its length is $n_1+3y+3 \leqslant  n_1+3x$. Then either $n_1+3x \leqslant \frac34n + 4$, in which case we conclude, or $n_1+3x > \frac34n + 4$. In the latter case we can rewrite inequality \eqref{eq:three} as follows: $$l \leqslant  n_{\short} - y + 1 + a = n - n_1 - 4x - y + 5 = n - (n_1+3x) -x-y+5 \leqslant  \frac{n}4 - x -y+1,$$ which implies inequality ($3'$). If there are at most four blocks of 0's in $v$ which have lengths at least $\frac{y}2$, then we can consider the subword of the word $u$ and not of $v$ which has $n_1$ occurrences of 1 and five blocks $0^{\frac{y}2}$. So, we can get  ($3'$) similarly to the previous case, where $v$ has at least six blocks of 0's which have lengths at least $\frac{y}2$.

Now we assume that both $u$ and $v$ have five blocks of 0's with lengths at least $\frac{y}2$. We call them {\it major} blocks. Now we prove the following claim:

\begin{claim}
Let $u$ be not a special word with $a_u = 4$. Then there is a word $u'$ with $n_1$ occurrences of 1 and two blocks $0^x$ such that there is at least one and are at most two occurrences of $u'$ in $u$.    
\end{claim}

\begin{proof}
By the conditions of the claim $u_{\longg}$ has the following form: $u_{\longg} = 0^x1^{\alpha} 0^x1^{\beta} 0^x 1^{\gamma} 0^x 1^{\delta}$ for some positive integers $\alpha, \beta, \gamma, \delta$. Denote $t = min(\alpha, \beta, \gamma, \delta)$; it is obvious that $t < n_1/2$. If among the numbers $\alpha, \beta, \gamma, \delta$ at most two are equal to $t$,  then there are at most two ways to take an occurrence of the subword  $u' = 0^x1^t0^x1^{n_1-t}$ in $u$. In the other case $u = 0^x1^t0^x1^t0^x1^t0^x1^{t_0}$. Since $u$ is not special, we have $t_0 \neq t$, $t_0 \neq 2t$. Since $t_0 > t$, we have $2t < \frac{n_1}2$. In this case we can take $u' = 0^x1^{2t}0^x1^{n_1-2t}$, and  there are at most two ways to take an occurrence of the subword $u' = 0^x1^{2t}0^x1^{n_1-2t}$ in $u$. The claim is proved. \qedSymC 
\end{proof}

By the claim there is a subword $u'$ of $u$ with $n_1$ occurrences of 1 and two blocks $0^x$ such that there is at least one and there are at most two subwords of $u$ which are equal to $u'$. We let $u_1$ and $u_2$ denote these subwords, which are equal to $u'$ (if $u_2$ exists). Then either $v$ does not have a subword which is equal to $u'$ or there is a subword $v' = u'$ of $v$. In the first case $u'$ is the subword of the word $u$ and is not a subword of $v$, and $|u'| = n_1+2x \leqslant  n_1+\frac{n_0}2 \leqslant  \frac34n$. In the second case we consider two 1-overlays of the word $v$ on $u$ such that $v'$ and $u_1$ (resp., $v'$ and $u_2$)  coincide. For both overlays we have that since $v \neq u$, there is a block $0^{s_1}$ (resp, $0^{s_2}$) in $u$ (possibly empty) which is smaller than the corresponding block of 0's in $v$ (since $s \leqslant  y$). We add a block $0^{s_1+1}$ ($0^{s_2+1}$) to $v'$ in the corresponding place. So, we add at most two blocks. Note that for one of the overlays $0^{s_i} \neq 0^y$; otherwise in both overlays the major blocks of $v$ coincide with the major blocks of $u$. So, there is a turn of $u$ which translates the major blocks of $u$ to the major blocks of $u$. Then the subword of $u$ which contains $n_1$ occurrences of 1, all big blocks $0^x$ and block $0^y$ is equal to $0^x1^t0^x1^t0^x1^t0^x1^t0^y1^t$. That is $u_{\longg} = 0^x1^t0^x1^t0^x1^t0^x1^{2t}$ is special; a contradiction. So after adding blocks $0^{s_1}$ and $0^{s_2}$ (one of them contains less than $\frac{y}2$ occurrences of 0) we get the word $v''$ such that $v''$ is a subword of the word $v$ and is not a subword of $u$ and $|v''| \leqslant  n_1+2x+y+\frac{y}2+2$. Then either $v''$ is a distinguishing subword of a desired length, or $n_1+2x+y+\frac{y}{2} \geqslant  \frac34n+2$. But if $n_1+2x+y+\frac{y}{2} \geqslant  \frac34n+2$, then $$l \leqslant  n - n_1 - 4x - y + 5 = n - (n_1+2x+y+\frac{y}2) - 2x+\frac{y}2+5 \leqslant  \frac{n}4 - 2x -\frac{y}{2}+5,$$ so $l \leqslant \frac{n}4 - x -\frac{y}{2}+3$, which is what we needed to get.

\medskip

\noindent {\bf Case 3.} $a = 3$.

\medskip
In this case $u_{\fix}$ contains two blocks $0^x$. So, we can prove the lemma in this case similarly to Cases {\bf 1} and {\bf 2}, but instead of Case {\bf 1} we consider Case ${\bf 1'}$:  $n_1+2x + y \leqslant  \frac34n+3$, and instead of Case {\bf 2} we consider Case {$\bf 2'$}: $n_1+2x + y \geqslant  \frac34n+3$. 

In Case ${\bf 1'}$ the subword $u_{\fix}$ of $u$ is unioccurrent and $|u_{\fix}|+y+1 \leqslant  n_1 + 2x+y+1 \leqslant   \frac34n+4$. Then applying Proposition~\ref{prop popular idea} to the word $u_{\fix}$, we get that there exists a distinguishing subword of length at most $\frac34n+4$ for $u$ and $v$.

In Case ${\bf 2'}$, since $a = 3$, we can rewrite inequality \eqref{eq:three} as follows:

$$l \leqslant  n_{\short} - y + 1 + a = n-n_1-3x-y+a+1 = $$
\begin{equation} \tag{$3''$}\label{eq1}
= n - (n_1 + 2x + y+1) - x + a \leqslant  \\ \frac{n}4 - x + a-3 = \frac{n}4 - \frac{(a-1)(x-2)}2-2. 
\end{equation}
Using ($3''$), we can rewrite inequality \eqref{eq:four} in Case {\bf 2.1.1} as  $$|v_{\zeros}|+y+1= n_1 + l + y+1 \leqslant  n_1 + \frac{n}4 - \frac{(a-1)(x-2)}{2}-2 + x =  $$ $$=  \frac{n}2 + \frac{n}4 - \frac{ax-3x-2a+6}{2} =  \frac34n - \frac{(a-3)(x-2)}{2} \leqslant  \frac34n + 2,$$ 
since $a \geqslant  3$ and $x \geqslant  2$. Using ($3''$), we can rewrite inequality \eqref{eq:five} in Case {\bf 2.1.2} as: $$|v_{\zeros}| \leqslant  \frac{n_1+l}2 + l + 2(x-1)+y =  \frac{n_1+3x+y+(l-4)}2 + l + \frac{x+y}2\leqslant $$ $$\leqslant  \frac{n}2+\frac{n}4 - \frac{(a-1)(x-2)}2 - 2 + x =  \frac34n - \frac{(a-3)(x-2)}2,$$
since $a \geqslant  3$ and $x \geqslant  2$.  \qed

\begin{lemma} \label{lemma 3}
    Let $u_{\longg}$ be periodic and not special with $a \geqslant  3$. Then there exists a 
    distinguishing subword of length at most $\frac34n+4$.
\end{lemma} 

\textit{Proof.} Since for $a=3$ we have that $u_{long}$ is special as it is periodic with exactly three blocks $0^x$, we only have to look at the cases when $a\geq 4$.

\medskip
\noindent {\bf Case 1.} $a \geqslant  5$.
\medskip

In this case $u_{\longg} = (1^{\alpha_1}0^x1^{\alpha_2}0^x\cdots 1^{\alpha_r}0^x)^t$ for some $r$ and $t$ such that $rt = a$ and $t \geqslant  2$. So, $r \leqslant  a/2$. Since $u_{\longg}$ is not special, we have $r \geqslant  2$. Consider the word $u_1 =(1^{\alpha_1}0^x1^{\alpha_2}0^x\cdots 1^{\alpha_r}0^x)1^{r(t-1)}$. Then either $u_1$ is not a subword of $v$, or there is a subword $v_1 = u_1$ of $v$. In the first case $u_1$ is a subword of the word $u$ and is not a subword of $v$, and $|u_1| = n_1+\frac{ax}{t} \leqslant n_1 + \frac{n_0}2 \leqslant \frac34n$. In the second case there are at most $t$ 1-overlays of $v$ on $u$ such that $v_1$ coincides with a an occurrence of a subword equal to $v_1$. For each such 1-overlay there is a block $0^s$ $(s \leqslant  y)$ in $u$ which is smaller than the corresponding block in $v$. We add $0^{s+1}$ to $v_1$ to the corresponding place; let $v_1'$ denote the obtained subword of $v$. The word $v_1'$ is a subword of the word $v$ and is not a subword of $u$. Since $y+1 \leqslant x$, we have $$|v_1'| \leqslant n_1 + rx+\frac{a}{r}(y+1) \leqslant n_1+rx+\left(\frac{a}{r}-1\right)x+(y+1).$$ 

Consider the function $f(r)= r+\frac{a}{r}-1$. It is easy to see that for $a \geqslant 5$, in the interval $2 \leqslant  r \leqslant  a/2$ we have $f(r) \leqslant \max(f(2), f(\frac{a}2)) = \frac{a}{2}+1$.  Then 
$$|v_1'| \leqslant n_1+rx+\left(\frac{a}{r}-1\right)x+(y+1) \leqslant n_1+\frac{ax}2+x+y+1.$$
So, if $n_1 + \frac{ax}2+x+y \leqslant  \frac34n+3$, then we can take $v'$ as a  required subword. 

Assume that $n_1 + \frac{ax}2+x+y \geqslant  \frac34n+3$.  Similarly to Case {\bf 2} from Lemma~\ref{lemma 2} we can obtain inequality \eqref{eq:three}: $l \leqslant \frac{n}4 - \frac{(a-2)(x-2)}{2}$. Now we proceed with the proof similarly to Case {\bf 2} from Lemma~\ref{lemma 2}.

If $u_{\zeros}$ is not periodic, then the proof is similar to Case {\bf 2.1.1} from Lemma~\ref{lemma 2}. Assume that $u_{\zeros}$ is periodic. That is, $u_{\zeros} = (1^{\alpha_1}01^{\alpha_2}0\cdots 1^{\alpha_p}0)^\frac{l}p$, for some $p\geq 2$. Consider the word $u_{\zeros}' = (1^{\alpha_1}01^{\alpha_2}0\ldots 1^{\alpha_p}0)(10)^{l-p}$. Similarly to Case {\bf 2.1.2} from Lemma~\ref{lemma 2}, the word $u_{\zeros}'$ contains $l$ occurrences of 0 and at most $\frac{n_1+l}2$ occurrences of 1. 
We add to $u_{\zeros}'$ all blocks $0^x$ from $u_1$ (recall that $u_1 = (1^{\alpha_1}0^x1^{\alpha_2}0^x\ldots 1^{\alpha_r}0^x)1^{r(t-1)}$). We let $u_{\zeros}''$ denote the obtained subword of $u$. So either $u_{\zeros}''$ is not a subword of $v$ and $|u_{\zeros}''| \leqslant \frac34n+4$ (we will prove this inequality later), or there is a subword $v' = u_{\zeros}''$ of $v$. In the latter case there are at most $t$ 1-overlays of $v$ on $u$ such that $v'$ coincides with an occurrence of a subword equal to $v'$. For each of these overlays there is a block $0^s$ $(s \leqslant  y)$ in $u$ which is smaller than the corresponding block in $v$. We add $0^{s+1}$ to $v'$ to the corresponding place. We let $v''$ denote the obtained subword of $v$. The word $v''$ is a subword of the word $v$ and is not a subword of $u$ and it has at most $\frac {n_1+l}{2}$ occurrences of 1, one zero from each block of 0's, $rx$ occurrences of 0 from big blocks $0^x$ and $t$ added blocks each of which contains at most $y+1$ occurrences of 0. Since $rt = a$ and $r,t \geqslant 2$, we have $r+t \leqslant  \frac{a+2}{2}+1$. Since $r+t \leqslant  \frac{a+2}{2}+1$, $y \leqslant x-1$, we have 
$$|v''| \leqslant  \frac{n_1+l}{2} + l + r(x-1) + ty \leqslant  \frac{n_1+l}2 + l + \frac{(a+2)}2(x-1)+y \leqslant   \frac34n+4.$$ The last inequality can be proved similarly to inequality \eqref{eq:five} from Lemma~\ref{lemma 2}.

\medskip

\noindent {\bf Case 2.} $a = 4$.

\medskip

Since $u_{\longg}$ is not special, we have $u_{\longg} = 0^x1^{\alpha}0^x1^{\beta} 0^x1^{\alpha}0^x1^{\beta}$, $\alpha \neq \beta$. This case is proved in the same way as Case {\bf 2.2} from Lemma~\ref{lemma 2}. \qed





\begin{lemma} \label{lemma 4}
    Suppose $u_{\longg}$ is not a special word with at most two blocks $0^x$ and  $v_{\longg}$ is either special or contains at most two blocks $0^x$ (or vice versa).
    Then there is a distinguishing subword of length at most $\frac34n+4$.
\end{lemma}

\textit{Proof.} Consider four cases corresponding to possible values of $a$ and $b$.

\medskip

\noindent {\bf Case 1.} $(a,b) = (1,1)$. 

\medskip

Recall that $0^{y_u} = 0^y$ and $0^{y_v}$ are the second largest blocks of 0's in $u$ and $v$, respectively. Without loss of generality we may assume that $y \geqslant  y_v$. We now consider a 1-overlay of $u$ on $v$ such that the blocks $0^x$ coincide; we let $\pi$ denote this overlay. If for this 1-overlay $\pi$ there are two neighboring 1's such that there are 0's between them in only one of the words $u$ and $v$, then the word $v_1 = 0^{y+1}1^{\alpha}01^{n_1-\alpha}$ (or $u_1 = 0^{y_v+1}1^{\alpha}01^{n_1-\alpha}$) for some $\alpha$ is a distinguishing subword for the words $u$ and $v$. Moreover, the length of this word is at most $n_1+y+2 \leqslant  n_1 + (x+y)/2 + 2 \leqslant  n_1 + n_0/2+2 \leqslant  \frac34n+2$. So, in this case we have a required subword.

Otherwise, for a 1-overlay $\pi$ for each place where $u$ has a block of 0's there is a block of 0's in $v$ and vice versa. In this case, there are two blocks $0^{\beta}$ and $0^{\gamma}$ in $u$ such that $0^{\beta}$ is bigger than the corresponding block in $v$ and $0^{\gamma}$ is smaller than the corresponding block in $v$. Without loss of generality we assume that $\beta \leqslant  \gamma$ (the case $\beta > \gamma$ is similar). Then for some integer $\alpha$ the word $u_1 = 0^{y+1}1^{\alpha}0^{\beta}1^{n_1-\alpha}$  is a subword of the word $u$ and is not a subword of $v$ and $|u_1| = n_1+(y+1)+\beta = n_1+n_0-x-(n_{\short}-y-\beta)+1 \leqslant  n - x - l + 4$, since $n_{\short}-y-\beta \geqslant  l - 3$. So, if $n - x - l + 4 \leqslant \frac34n+4$, then we have a required subword. If $n - x - l + 4 > \frac34n+4$, then we get $x+l < \frac{n}4$. Since for a 1-overlay $\pi$ for each place where $u$ has a block of 0's there is a block of 0's in $v$ and vice versa, we have that $u_2 = 0^{y+1}1(01)^{\alpha-1}0^{\beta}1(01)^{l-\alpha-1}$ is a subword of the word $u$ and is not a subword of $v$ and $|u_2| = 2l + y + \beta - 1 < 2l+2x < \frac{n}2$. We proved the lemma in Case {\bf 1}.

\medskip

\noindent {\bf Case 2.} $(a,b) = (2,1)$. 

\medskip

In this case the word $u_1 = 0^{x}1^{\alpha_1}0^{x}1^{\alpha_2}$ is a subword of the word $u$ and is not a subword of $v$ and $|u_1| = n - n_{\short}$. If $n - n_{\short}\leqslant \frac{3}{4}n$, then we have a required subword, so it remains to consider the case $n_{\short} < \frac{n}4 - 4$. This inequality implies that $l \leqslant  n_{\short}-y+3 <  \frac{n}4 - y - 1$. The word $u_2 = 0^{x}1(01)^{\alpha_1-1}0^{x}1(01)^{\alpha_2-1}$ is a subword of the word $u$ and is not a subword of $v$, and $|u_2| = 2l + 2x -2$. If $2l + 2x -2\leqslant \frac{3}{4}n$, then we again have a required subword, so it remains to consider the case $2l+2x-2 > \frac34n+4$. The latter inequality implies that $2(\frac{n}4 - y - 1) + 2x - 2 >  \frac34n +4$, which can be rewritten as $2x-2y \geqslant  \frac{n}4 + 8$. Now consider the subword $v_3 = 1^{n_1}0^{y+1}$ of $v$. There are two 1-overlays of $v$ on $u$ such that $v_3$ coincides with an occurrence of an equal subword. For both overlays we find blocks $0^{s_1}$ and $0^{s_2}$ which are shorter than the corresponding blocks in $v$, and we add to $v_3$ blocks $0^{s_1+1}$ and $0^{s_2+1}$ to the corresponding places; we let $v_3'$ denote the obtained subword of $v$. The word $v_3'$ is a subword of the word $v$ and is not a subword of $u$, and $$|v_3'| \leqslant  n_1 + 3(y+1) = n_1 + y + 2x + 2(y-x) + 3 \leqslant  n - \frac{n}4-5,$$ since $n_1+y+2x \leqslant  n$ and $2(y-x) \leqslant  -\frac{n}4 - 8$. This completes the proof in Case {\bf 2}.

\medskip

\noindent {\bf Case 3.} $(a,b) = (2,2)$.

\medskip

Let $u_{\longg} = 0^{x}1^{\alpha_1}0^{x}1^{\alpha_2}$ and $v_{\longg} = 0^{x}1^{\beta_1}0^{x}1^{\beta_2}$, where $\alpha_1+\alpha_2 = \beta_1+\beta_2 = n_1$. Assume without loss of generality that $\alpha_1 \leqslant \beta_1 \leqslant \beta_2 \leqslant \alpha_2$. Let $$u = 0^x\,\, 1^{r_1}0^{\gamma_1}\cdots 0^{\gamma_{i-1}}1^{r_i}\,\,0^x\,\, 1^{r_{i+1}}0^{\gamma_{i}}\cdots 0^{\gamma_{l-2}}1^{r_{l}}$$
$$v = 0^x\,\, 1^{p_1}0^{\phi_1}\cdots 0^{\phi_{j-1}}1^{p_j}\,\,0^x\,\, 1^{p_{j+1}}0^{\phi_{j}}\cdots 0^{\phi_{l-2}}1^{p_{l}},$$
where $r_1+\ldots+r_i = \alpha_2$, $r_{i+1}+\ldots+r_l = \alpha_1$, $p_1+\ldots+p_j = \beta_2$, $p_{j+1}+\ldots+p_l = \beta_1$. Recall that $y_u$ and $y_v$ are the lengths of the second longest blocks of 0's in $u$ and $v$, i.e. the largest length smaller than $x$. Let $y_{max} = \max(y_u, y_v)$. Consider two cases.

\medskip

\noindent {\bf Case 3.1.} $\beta_1 > \alpha_1$. 

\medskip


Since $\beta_1 > \alpha_1$, we have $\alpha_1 < \beta_1 \leqslant \beta_2 < \alpha_2$. Then the word $u' = 0^{y_{max}+1}1^{\alpha_1}0^{y_{max}+1}1^{\alpha_2}$ is a subword of the word $u$ and is not a subword of $v$, and $|u'| \leqslant n_1+2y_{max}+2$. Then either we have a required subword, or $n_1+2y_{max} > \frac34n+2$. Now we assume that $n_1+2y_{max} > \frac34n+2$. 

If $i \neq j$, then the word $u'' = 0^{y_{max}+1}1(01)^{i-1}0^{y_{max}+1}(01)^{l-i}$ is a subword of the word $u$ and is not a subword of $v$ and $|u''| \leqslant 2l+y_{max}$. If $2l+y_{max} > \frac34n+4$, then $$(2l+2y_{max}) + 2(n_1+2y_{max}) > \frac34n+4 + 2\left(\frac34n+2\right),$$ which implies $2n_1+2l+6y_{max} > \frac94n+8$. Since $n_0 > 2x+y_{max}+(l-3)$, we have $$2n = 2n_1 + 2n_0 > 2n_1 + 2(2x+y_{max}+(l-3)) > 2n_1+2l+6y_{max}-6 > \frac94n+2;$$ a contradiction. So, $|u'| \leqslant \frac34n+4$, and hence in the case $i \neq j$ we have a required subword. Now assume that $i = j$.

Since $\alpha_1 < \beta_1$ and $\alpha_2 > \beta_2$, there are indices $t_1 \leqslant i$ and $t_2 \leqslant l-i$ such that $r_{t_1} < p_{t_1}$ and $r_{i+t_2} > p_{i+t_2}$. Assume without loss of generality that $r_{t_1} < r_{i+t_2}$. Consider the subword $v' = 0^{y_{max}+1}(10)^{t_1-1}1^{r_{t_1}+1}(01)^{l-t_1}$ of $v$. Then there is at most one 1-overlay of $v$ on $u$ such that $v'$ coincides with an occurrence of an equal word (modulo selection of 0's and 1's from each block). Since $u \neq v$, for this 1-overlay either there is a block $0^{s_1}$ in $u$ which is smaller than the  corresponding block in $v$ ($s_1 \leqslant y_{max}$) or there is a block $1^{s_2}$ in $u$ which is bigger than the corresponding block in $v$. In the first case we add to $v'$ the block $0^{s_1+1}$ to the corresponding place; we let $v_1$ denote the obtained subword of $v$. In the second case we add to $v'$ the block $1^{s_2}$ to the corresponding place; we let $u_1$ denote the obtained subword of $u$. The words $v_1$ and $u_1$ are distinguishing subwords for $u$ and $v$.
We have that $|v_1| \leqslant 2l+2y_{max}+r_{t_1}$ and $|u_1| \leqslant n_1+l+y_{max}$. 

If $n_1+l+y_{max} > \frac34n+4$, then $$(n_1+l+y_{max})+(n_1+2y_{max}) > \frac34n+4+\frac34n+2 = \frac32n+6.$$ However, $$\frac32n > \frac{n}2+n\geqslant n_1 + (n_1+2x+y+(l-3)) \geqslant 2n_1 + 3y_{max} + l - 1 \geqslant \frac32n+5;$$ a contradiction. Thus $|u_1| \leqslant \frac34n+4$. 

Since $(r_{t_1}-1) + (r_{i+t_2}-1) \leqslant n_1-l$, we have $$|v_1| \leqslant 2l+2y_{max}+\frac{r_{t_1}+r_{i+t_2}}2 \leqslant 2l+2y_{max}+\frac{n_1-l}2+1 = \frac{n_1}2+\frac32l+2y_{max}+1.$$ 
If $\frac{n_1}2+\frac32l+2y_{max} > \frac34n+3$, then $$\left(\frac{n_1}2+\frac32l+2y_{max}\right)+(n_1+2y_{max}) > \left(\frac34n+3\right) + \left(\frac34n+2\right) = \frac32n+5.$$ We have that $$\left(\frac{n_1}2+\frac32l+2y_{max}\right)+(n_1+2y_{max}) < \frac32(n_1+2x+y_{max}+(l-3))+\frac92 < \frac32n+5;$$ a contradiction. Thus $|v_1| \leqslant \frac34n+4$.

Since for the lengths of the obtained subwords in both cases we have $|u_1| \leqslant \frac34n+4$ and $|v_1| \leqslant \frac34n+4$, we have a proof of the lemma in this case. 

\medskip

\noindent {\bf Case 3.2.} $\beta_1 = \alpha_1$. 

\medskip

Consider a 1-overlay of $v$ on $u$ such that $u_{\longg}$ and $v_{\longg}$ coincide. There is only one such overlay, since $\alpha_1 < \frac{n_1}2$ ($u_{\longg}$ is not special). Since $u \neq v$, there are two blocks $0^s$ and $0^t$ in $u$ such that $0^s$ is bigger than the corresponding block in $v$ and $0^t$ (possible empty) is smaller than the corresponding block in $v$ ($s, t \leqslant  y$). Without loss of generality assume that $s > t$ (in the case $s \leqslant  t$ the proof is similar). We now consider the word $v_1 = 1^{n_1}$ and add $y_{max}+1$ zeros from each of the block $0^x$ to $v_1$ to the corresponding places. We also add $t+1$ zeros to $v_1$ from the block which is bigger than the corresponding block $0^t$ in $u$. We let $v_1'$ denote the obtained subword of $v$. The word $v_1'$ is a subword of the word $v$ and is not a subword of $u$, and $|v_1'| = n_1+2(y_{max}+1)+(t+1)$. If $n_1+2y_{max}+t \leqslant  \frac34n + 2$, then we have a required subword, so it remains to consider the case $n_1+2y_{max}+t >  \frac34n + 2$. 

Consider the subword $v_2 = 0^{y_{max}+1}1(01)^{l-1}$ of $v$. There are two 1-overlays of $v$ on $u$ such that $v_2$ coincides with an occurrence of an equal subword (modulo selection of 0's and 1's from each block). For each 1-overlay there exists either a block of 0's or a block of 1's in $u$ which is smaller than the corresponding block in $v$ and vice versa. Consider three cases.

\smallskip

\noindent {\bf Case 3.2.1.} If for both 1-overlays there are blocks of 1's in $v$ which are bigger than the corresponding blocks of 1's in $u$, then we add them to $v_2$. We let $v_2'$ denote the obtained subword  of $v$. The word $v_2'$ is a subword of the word $v$ and is not a subword of $u$, and $|v_2'| \leqslant n_1+l+y_{max}$, since $v_2'$ has at most $n_1$ occurrences of 1 and $l+y_{max}$ occurrences of 0. If $n_1+l+y_{max} >  \frac34n + 4$, then $$\frac32n = \frac{n}2+n > n_1+(n_1+2x+y+t+(l-4)) > 2n_1+3y_{max}+t+l-2 =$$ $$(n_1+l+y_{max})+ (n_1+2y_{max}+t) - 2 > \left(\frac34n + 4\right)+\left(\frac34n + 2\right) - 2 = \frac32n+4;$$ a contradiction. Thus $|v_2'| < \frac34n+4$. In this case the lemma is proved.

\medskip

\noindent {\bf Case 3.2.2.} If for both 1-overlays there are blocks $0^{s_1}$ and $0^{s_2}$ in $u$ which are smaller than the corresponding blocks of 0's in $v$, then we add blocks $0^{s_1+1}$ and $0^{s_2+1}$ to $v_2$ to the corresponding places. We let $v_2''$ denote the obtained subword of $v$. The word $v_2''$ is a subword of the word $v$ and is not a subword of $u$ and $|v_2''| \leqslant 2l+3y_{max}$ since $|v_2| = 2l+y_{max}$ and we add $s_1+s_2 \leqslant 2y_{max}$ occurrences of 0. If $2l+3y_{max} >  \frac34n + 4$, then $$2n + y_{max} = 2(n_1+2x+y+t+(l-4))+y_{max} > 2n_1+7y_{max}+2t+2l-4 = $$ $$(2l+3y_{max}) + 2(n_1+2y_{max}+t) - 4 > \left(\frac34n + 4\right)+2\left(\frac34n + 2\right)-4= \frac94n+4,$$ which implies $y_{max} > \frac{n}4+4$. Then $$n > (n_1+2y_{max}+t)+y_{max} > \left(\frac34n + 2\right) + \left(\frac{n}4+4\right) = n +6;$$ a contradiction. Thus $|v_2''| < \frac34n+4$. In this case the lemma is proved.

\medskip
 
\noindent {\bf Case 3.2.3.} If for one 1-overlay there is a block of 1's in $v$ which is bigger than the corresponding block of 1's in $u$ and for the other 1-overlay there is a block of 0's in $u$ which is smaller than the corresponding block of 0's in $v$, then for the first 1-overlay there are blocks $1^{t_1}$ and $1^{t_2}$ in $u$ such that $1^{t_1}$ is smaller than the corresponding block in $v$ and $1^{t_2}$ is bigger than the corresponding block in $v$. Assume without loss of generality that $t_1 < t_2$. Let $0^{s}$ be the block of 0's in $u$ which is smaller than the corresponding block in $v$ for the second 1-overlay. We add $t_1$ occurrences of 1 and $s$ occurrences of 0 to $v_2$. We let $v_2'''$ denote the obtained subword word of $v$. The word $v_2'''$ is a subword of the word $v$ and is not a subword of $u$, and $|v_2'''| \leqslant 2l+\frac{n_1-l}2+2y_{max}$, since $|v_2| = 2l+y_{max}$ and we added at most $\frac{n_1-l}2$ occurrences of 1 and and at most $y_{max}$ occurrences of 0. If $2l+\frac{n_1-l}2+2y_{max} > \frac34n+4$, then $$\frac32n = \frac32(n_1+2x+y+t+(l-4)) > (n_1+2y_{max}+t)+\left(2l+\frac{n_1-l}2+2y_{max}\right)-3 > $$ $$\left(\frac34n + 2\right) + \left(\frac34n+4\right) - 3 = \frac32n+3;$$ a contradiction. Thus $|v_2'''| < \frac34n+4$. In this case the lemma is proved.

\medskip

\noindent {\bf Case 4.} $(a,b) = (2,\geqslant  3)$. 

\medskip

In this case we have $u_{\longg} = 0^{x}1^{\alpha_1}0^{x}1^{\alpha_2}$ and $v_{\longg} = 0^{x}1^{\beta_1}0^{x}1^{\beta_2}\cdots 0^x1^{\beta_b}$, for some integers $x$, $\beta_1$, \dots, $\beta_b$ and $\alpha_1 \leqslant \alpha_2$. Since $u_{\longg}$ is not special and $b \geqslant 3$, we have $\alpha_2 > \frac{n}2 > \alpha_1$ and $\beta_i \leqslant \frac{n}2$ for each $1 \leqslant i \leqslant b$. If there is $\beta_i \neq \alpha_1$ then the proof is similar to the proof in Case {\bf 3.1} for the words $u_{\longg} = 0^{x}1^{\alpha_1}0^{x}1^{\alpha_2}$ and $v_{\longg} = 0^{x}1^{\beta_i}0^{x}1^{n_1-\beta_i}$, since $(\alpha_1,\alpha_2) \neq (\beta_i, n_1-\beta_i)$. Assume that $\beta_i = \alpha_1$ for each $i$. Since $u$ is not special, we have $\alpha_1 \neq \frac{n_1}3$. Then $b \geqslant 4$. So, we can prove the lemma as in the Case {\bf 3.1} for the words $u' = u_{\longg} = 0^{x}1^{\alpha_1}0^{x}1^{\alpha_2}$ and $v' = 0^{x}1^{2\beta_1}0^{x}1^{n_1-2\beta_1}$ since $(\alpha_1,\alpha_2) \neq (2\beta_1, n_1-2\beta_1)$.  \qed 

\begin{lemma} \label{lemma 5}
    Let $u_{\longg}$ and $v_{\longg}$ be special words. Then there is a distinguishing subword  of length at most $\frac34n+4$ for $u$ and $v$.
\end{lemma}

\textit{Proof.} We recall that a special word has at least two blocks of 0's, that is, $a \geqslant  2$ and $b \geqslant  2$, and that  $dist(u_{\longg})$ denotes the length of the shortest block of 1's in a special word. 

First we prove this lemma in the case  $dist(u_{\longg}) \neq dist(v_{\longg})$:

\medskip

\noindent {\bf Case 1.} $dist(u_{\longg}) \neq dist(v_{\longg})$. 

\medskip

 The proof in the case $a = b = 2$ is similar to the proof in Case {\bf 3.1} of Lemma~\ref{lemma 4}. Now assume that $a \geqslant  3$.

\medskip

\noindent {\bf Case 1.1.} $a \geqslant  3$ and $b \geqslant  3$. 

\medskip

Assume that $dist(u_{\longg}) > dist(v_{\longg})$. In this case we can take the following word $v_1$, which is a subword of $v$ and not a subword of $u$: the word $v_1$ contains $n_1$ occurrences of 1 and two blocks $0^{y+1}$ at distance   $dist(v_{\longg})$. Since $a \geqslant  3$, we have $$|v_1| = n_1 + 2(y+1) \leqslant  n_1 + \frac{ax+y}{2}+2 \leqslant  \frac{n_1+n_0}{2} + \frac{n_1}{2} + 2 \leqslant  \frac{n}2 + \frac{n}{4} + 2\leqslant  \frac34n + 2.$$

\medskip

\noindent {\bf Case 1.2.} $a \geqslant  3$ and $b = 2$ (or, symmetrically, $a = 2$ and $b \geqslant  3$). 
\medskip

Without loss of generality, consider the case $a \geqslant  3$ and $b = 2$. In this case $dist(u_{\longg}) < dist(v_{\longg})$, since $dist(u_{\longg}) \leqslant  \frac{n_1}3$ and $dist(v_{\longg}) \geqslant  \frac{n_1}3$  (the latter inequality comes from the fact that for $b=2$, we have either $v_{\longg}=0^x1^\frac{n_1}{2}0^x1^\frac{n_1}{2}$ or $v_{\longg}=0^x1^\frac{n_1}{3}0^x1^\frac{2n_1}{3}$, which gives in the first case $dist(v_{\longg}) = \frac{n_1}{2}$ and in the second case $dist(v_{\longg}) = \frac{n_1}{3}$). Consider two subwords of $u$: a subword $u_1$ containing $n_1$ occurrences of 1 and two blocks $0^x$ at distance $dist(u_{\longg})$ and a subword $u_2$ containing one 0 and one 1 from each block and three blocks $0^x$. Both words $u_1$ and $u_2$ are subwords of the word $u$ and are not subwords of $v$. So either we have a distinguishing subword, or $|u_1| = n_1+2x \geqslant  \frac34n+5$ and $|u_2| = 2l+3(x-1) \geqslant  \frac34n+5$. Summing the second inequality with the tripled first inequality we get: $$3(n_1+2x)+(2l+3x-3) \geqslant  3n+20.$$ 
However, we have $$3n +20\geqslant  3(n_1+3x+(l-3))+20 \geqslant  3n_1+9x+2l+11,$$ so we have a contradiction.

\medskip
From now on we assume that $dist(u_{\longg}) = dist(v_{\longg}) = k$. Then $u_{\longg}$ and $v_{\longg}$ are of the form 

$$0^x1^k \cdots 0^x1^k(0^x|\varepsilon)1^k0^x1^k\cdots  0^x1^k(0^x|\varepsilon)1^k0^x1^k \cdots 0^x1^k,$$
where the notation $(0^x|\varepsilon)$ is the standard notation for regular expressions meaning that we either take $0^x$ or the empty word $\varepsilon$. So, if $u_{\longg}$ or $v_{\longg}$ is of type 1, then we choose $0^x$ in both places; if $u_{\longg}$ or $v_{\longg}$ is of type 2, then we choose $0^x$ in one of the places and $\varepsilon$ in the other place; if $u_{\longg}$ or $v_{\longg}$ is of type 3, then we choose $\varepsilon$ in both places.
We now add to $u_{\longg}$ and $v_{\longg}$ blocks of 0's from $u$ and $v$ which are at positions of $(0^x|\varepsilon)$. These blocks are either $0^x$ (if we had $0^x$ in $u_{\longg}$ or analogously $v_{\longg}$), or they can be shorter  (if we had $\varepsilon$ in $u_{\longg}$ or analogously $v_{\longg}$). The new subwords are of the following form: 
$$0^x1^k \cdots 0^x1^k0^{\alpha_u}1^k0^x1^k\cdots  0^x1^k0^{\beta_u}1^k0^x1^k \cdots 0^x1^k,$$ 
$$0^x1^k \cdots 0^x1^k0^{\alpha_v}1^k0^x1^k\cdots  0^x1^k0^{\beta_v}1^k0^x1^k \cdots 0^x1^k,$$
where $0 \leqslant  \alpha_u,  \alpha_v, \beta_u,\beta_v \leqslant  x$; we denote them by  $u_{\longg}^{+\alpha, \beta}$ and $v_{\longg}^{+\alpha, \beta}$, respectively. So, $0 \leqslant  a - b \leqslant  2$. Let $u_{\short}^{-\alpha,\beta}$ be the subword of $u$ which contain $n_1$ occurrences of 1 and all small blocks (all blocks except $0^x$, $0^{\alpha}$ and $0^{\beta}$). In other words, $u_{\short}^{-\alpha,\beta}$ is obtained from $u_{\short}$ by removing two blocks $0^{\alpha_u}$ and $0^{\beta_u}$. The word $v_{\short}^{-\alpha,\beta}$ is defined analogously. 

Consider the subword $v' = 1^{n_1}0^{x}$ of $v$. We now consider a 1-overlay of $v$ on $u$ such that $v'$ coincides with an occurrence of an equal subword in $u$. Since $u$ and $v$ are not equal, for each such 1-overlay we can find and add to $v'$ a block of 0's in $v$ which is bigger than the corresponding block in $u$. Adding such blocks of 0's to $v'$, we get a subword $v''$ of $v$ which is not a subword of $u$, and $|v''| \leqslant  n_1+(a+1)x$. So, either we have a desired distinguishing subword, or $n_1+(a+1)x \geqslant  \frac34n+5$. So, in the rest of the proof we assume that 
\begin{equation}
    n_1+(a+1)x \geqslant  \frac34n+5 \quad \text{and} \quad n_1+(b+1)x \geqslant  \frac34n+5. \label{ineq: 1} 
\end{equation}

\medskip


We can now describe the plan of the rest of the proof. We distinguish between two main cases: Case {\bf 3} treats the situation when $u_{\short}^{-\alpha,\beta} \cong v_{\short}^{-\alpha,\beta}$ (here we use notation from Definition \ref{def_cong}) and $u_{\short}^{-\alpha,\beta}$ is periodic with a period containing $ks$ occurrences of 1 for some integer $s$; Case {\bf 2} treats the opposite case. Case {\bf 2} is divided into two subcases corresponding to whether inequality $l\geqslant  \frac{n}4-(b-2)(x-2)$ holds or not. Case {\bf 3} is divided into subcases {\bf 3.1} corresponding to $s \geqslant  2$, and {\bf 3.2} corresponding to $s = 1$. The latter case is further split into two subcases depending on whether $\alpha_u = \alpha_v$ or not (note that in the case of equality we also have $\beta_u = \beta_v$ due to the conditions of Case {\bf 3}).










\medskip

\noindent {\bf Case 2.} $u_{\short}^{-\alpha,\beta} \ncong v_{\short}^{-\alpha,\beta}$ or $u_{\short}^{-\alpha,\beta} \cong v_{\short}^{-\alpha,\beta}$ and $u_{\short}^{-\alpha,\beta}$ is either non-periodic or periodic, but there is no period with $ks$ blocks of 1's for some integer $s$. 
\medskip

Without loss of generality we assume that $|u_{\short}^{-\alpha,\beta}| \leqslant  |v_{\short}^{-\alpha,\beta}|$ (the proof in the case $|u_{\short}^{-\alpha,\beta}| \geqslant  |v_{\short}^{-\alpha,\beta}|$  is symmetric). 

\medskip

\noindent {\bf Case 2.1.} $l \geqslant  \frac{n}4-(b-1)(x-2)+2$.

\medskip

Consider the subword $v_1 = 1^{n_1}0^x$ of $v$ and fix some of its occurrence in $v$. There are at most $a$ ways to take an occurrence of a subword of $u$ which is equal to $v_1$. For each of these $a$ occurrences we can consider a 1-overlay of $v$ on $u$ such that the chosen occurrence of $v_1$ in $v$  coincides with the chosen occurrence of $v_1$ in $u$.
Recall that we can assume that $dist(u_{\longg}^{+\alpha,\beta}) = dist(v_{\longg}^{+\alpha,\beta})$, since the case of inequality has been considered earlier (Case {\bf 1}). 

If $u_{\short}^{-\alpha,\beta} \ncong v_{\short}^{-\alpha,\beta}$, then for each 1-overlay there is a block of 0's in $v$ which is not one of big blocks $0^x$ or one of the two blocks $0^{\alpha}$ and $0^{\beta}$ we removed from $u_{\short}$, and which is bigger than the corresponding block in $u$. We now add all such blocks to $v_1$, and we let $v_1'$ denote the obtained subword of $v$. 

If $u_{\short}^{-\alpha,\beta} \cong v_{\short}^{-\alpha,\beta}$ and $u_{\short}^{-\alpha,\beta}$ does not have a period with $ks$ blocks of 1's, then the subwords $u_{\short}^{-\alpha,\beta}$ and $v_{\short}^{-\alpha,\beta}$ coincide only for one 1-overlay (if they coincide for more than one overlay, then they have a period with $ks$ blocks of 1's). For each 1-overlay for which they do not coincide, we add blocks of 0's to $v_1$ in the same way as in the case $u_{\short}^{-\alpha,\beta} \ncong v_{\short}^{-\alpha,\beta}$. Consider a 1-overlay for which $u_{\short}^{-\alpha,\beta}$ and $v_{\short}^{-\alpha,\beta}$ coincide. For this 1-overlay one of the blocks $0^{\alpha_u}$ and $0^{\beta_u}$ is smaller than the corresponding block in $v$. We then add a block $0^{\alpha_u+1}$ or $0^{\beta_u+1}$ to $v_1$ to the corresponding place, and we let $v_1'$ denote the obtained subword of $v$. 

The word $v_1'$ is the subword of the word $v$ and is not a subword of $u$. When we constructed $v_1'$, we did not add at least $l-b-a$ small blocks. So, we added to $v_1$ (of length $n_1+x$) at most $n_{\short,v} - (l-a-b)+1$ zeros. Indeed, if $u_{\short}^{-\alpha,\beta} \ncong v_{\short}^{-\alpha,\beta}$, then we added zeros to $v_1$ from only small blocks. If $u_{\short}^{-\alpha,\beta} \cong v_{\short}^{-\alpha,\beta}$, then $a-1$ times we added zeros from small blocks which are not equal to $0^{\alpha_v}$ and $0^{\beta_v}$, and one time we added the block $0^{\alpha_u+1}$ or $0^{\beta_u+1}$ which is less then $0^{\alpha_v+1}$ or $0^{\beta_v+1}$ since $\alpha_u+\beta_u = \alpha_v+\beta_v$. So,

$$|v_1'|  \leqslant 
(n_1 + x) + n_{\short,v} - (l-a-b-1).  $$

Since $2 \geqslant  a-b \geqslant  0$, we have 
$$(n_1 + x) + n_{\short,v} - (l-a-b-1) \leqslant n_1+n_{\short,v}-l+2b+x+3.  $$

Now since $n_1+n_{\short,v}+bx = |v| = n$, we have

$$ n_1+n_{\short,v}+bx-bx-l+2b+x+3 \leqslant   n - l - (b-1)(x-2)+5 \leqslant  \frac34n  + 3,$$ where the latter inequality comes from the condition $l \geqslant  \frac{n}4-(b-1)(x-2)+2$ of Case {\bf 2.1}. Combining this series of inequalities, we obtain $|v_1'|  \leqslant \frac34n  + 3.$

\medskip 

\noindent {\bf Case 2.2.} $l <  \frac{n}4-(b-1)(x-2)+2$.

\medskip

Consider the subword $u_{\zeros}$ of $u$, which contains $n_1$ occurrences of 1 and one 0 from each block of 0's, and fix some its occurrence in $u$. Then either $v$ does not have a subword which is equal to $u_{\zeros}$ or there is a subword $v_{\zeros} = u_{\zeros}$ of $v$. The first case is simple: the word $u_{\zeros}$ is a subword of $u$ and not a subword of $v$ and $|u_{\zeros}| \leqslant  \frac34n+4$ (we will prove this inequality later, together with the second case). In the second case we proceed depending on the form of the word $u_{\zeros}$. Fix some occurrence $v_{\zeros}$ in $v$.

If $u_{\zeros}$ is not periodic, then there is only one way to take $u_{\zeros}$ in $u$ modulo the selection of 0 from each block of 0's. We now consider the unique 1-overlay of $v$ on $u$ such that such that $v_{\zeros}$ and $u_{\zeros}$ coincide. Since $v \neq u$, there is a block $0^s$ in $v$ which is bigger than the corresponding block of 0's in $u$. We add this block to $v_{\zeros}$ to the corresponding place. Let $v'$ denote the obtained subword of $v$. The word $v'$ is a subword of the word $v$ and is not a subword of $u$. We now estimate its length: $$|v'| \leqslant |v_{zeros}| + (x-1) = n_1 + l + (x-1) $$
Now using the inequality  $l <  \frac{n}4-(b-1)(x-2)+2$ from the condition of Case {\bf 2.2}, we get
$$n_1 + l + (x-1)\leqslant n_1+\frac{n}{4}-(b-1)(x-2)+x+1 = n_1+\frac{n}4-(b-2)(x-2) + 3\leqslant  \frac34n+3,$$ where the last inequality follows from the inequalities $b \geqslant  2$, $x \geqslant  2$ and $n_1 \leqslant  \frac{n}2$. In particular, we have $|u_{\zeros}| \leqslant \frac34n+4$. 

If $u_{\zeros}$ is periodic, then $u_{\zeros} = (1^{\alpha_1}01^{\alpha_2}0\cdots 1^{\alpha_s}0)^\frac{l}s$ for some integer $s$. Consider the word $u_1 = (1^{\alpha_1}01^{\alpha_2}0\cdots 1^{\alpha_s}0)(10)^{l-s}$ and fix an occurrence of it in $u$. The word $u_1$ has at most $\frac{n+l}{2}$ occurrences of 1 and $l$ occurrences of 0 (proved in Lemma~\ref{lemma 2}, Case {\bf 2.1.2}). Then either $v$ does not have a subword which is equal to $u_1$ or there is a subword $v_1 = u_1$ of $v$. In the first case $u_1$ is a subword of the word $u$ and not a subword of $v$, and we  show that $|u_1| \leqslant  \frac34n+4$. In the second case we add one block $0^x$ to $v_1$. We let $v_1'$ denote the obtained subword of $v$ and fix an occurrence of it in $v$. Then there are at most $a$  1-overlays of $v$ on $u$ such that $v_1'$ coincides with an occurrence of an equal subword. For each 1-overlay we can add blocks as in Case {\bf 2.1} depending on whether $u_{\short}^{-\alpha,\beta} \cong v_{\short}^{-\alpha,\beta}$ or not. We let $v_1''$ denote the obtained subword of $v$. The word $v_1''$ has at most $\frac{n_1+l}{2}$ occurrences of $1$, one block $0^x$, one $0$ from each of the remaining $b-1$ blocks $0^x$ and at most $n_{\short}$ added 0's from other blocks of 0's. Now, by inequality \eqref{ineq: 1}, we have 
\begin{equation}
    n_{\short,v} = n_0-bx = (n_0+n_1+x) - (n_1+(b+1)x) \leqslant  n+x-\frac34n - 5 = \frac{n}4-5+x. \label{ineq: 2}
\end{equation} 
So, $v_1''$ is a subword of $v$ and is not a subword of $u$; we now estimate its length. Since $\frac{n_1+n_{\short,v}+bx}{2} = \frac{n}2$, by \eqref{ineq: 2} and due to inequality  $l <  \frac{n}4-(b-1)(x-2)+2$ from the condition of Case {\bf 2.2}, we have the following:
$$|v_1''|\leq \frac{n_1+l}2 + x + n_{\short,v}+b-1 \leqslant  
\frac{n_1+n_{\short,v}+bx}{2} - \frac{bx}2 + \frac{l}{2} + \frac{n_{\short,v}}2 +x+b-1 \leqslant $$ 

$$\leqslant  \frac{n}2 - \frac{bx}{2} + \frac{\frac{n}4-(b-1)(x-2)+2}{2} + \frac{\frac{n}4-5+x}2 + x + b-1 =$$
$$=\frac34n - \frac{2(b-2)(x-2)}{2}+\frac12 \leqslant  \frac34n+\frac12,$$ since $b \geqslant  2$ and $x \geqslant  2$. In particular, $|u_{\zeros}| \leqslant \frac34n+4.$

\bigskip 
\noindent {\bf Case 3.}  $u_{\short}^{-\alpha,\beta} \cong v_{\short}^{-\alpha,\beta}$ and $u_{\short}^{-\alpha,\beta}$ is periodic with period containing $ks$ blocks of 1's. In other words, there exist positive integers $p_1, \dots, p_i$ and $q_1, \dots, q_i$ such that $u_{\short}^{-\alpha,\beta} = (1^{p_1}0^{q_1}\cdots 1^{p_i}0^{q_i})^{\frac{n_1}{ks}}$, where $p_1+\ldots+p_i = ks$, for some integer $s$.

\medskip

Since $u_{\short}^{-\alpha,\beta} \cong v_{\short}^{-\alpha,\beta}$, $dist(u) = dist(v)$ and $n_{0,u} = n_{0,v}$, we have $\alpha_u+\beta_u = \alpha_v + \beta_v$. So, if one of the words $u_{\longg}$ and $v_{\longg}$ is a word of the first type, then the other 
word is a word of the first type as well. Then both words have period with $ks$ blocks of 1's, and since $u_{\short}^{-\alpha,\beta} \cong v_{\short}^{-\alpha,\beta}$, we have $u = v$, which gives a contradiction. Hence $u_{\longg}$ and $v_{\longg}$ are words of the second or the third type. In particular, $0 \leqslant  a - b \leqslant  1$. Consider two cases depending on $s$.

\medskip

\noindent {\bf Case 3.1.} $s \geqslant  2$. We consider two subcases depending on the types of the words $u_{\longg}$ and $v_{\longg}$. The proofs of the two cases are similar; however, we provide details for both.

\medskip

\noindent {\bf Case 3.1.1.} At least one of the words $u_{\longg}$ and $v_{\longg}$ is a word of the second type.

\medskip

Assume that $u_{\longg}$ is a word of the second type, that is, in $u_{\longg}^{+\alpha,\beta}$ we have $\alpha<k, \beta=k$. Then $n_1 = k(a+1)$. Consider the subword $u_1 = (1^{p_1}0^{q_1}\cdots 1^{p_i}0^{q_i})1^{n_1-ks}$ of $u$. There are at most $\frac{k(a+1)}{ks} \leqslant  \frac{a+1}{2}$ ways to choose an occurrence of a subword of $u_{\short}^{-\alpha,\beta}$ which is equal to $u_1$. Then either $v_{\short}^{-\alpha,\beta}$ does not have a subword which is equal to $u_1$ or there is a subword $v_1 = u_1$ of $v$. In the first case we add one block $0^x$ to $u_1$. Let $u_1'$ denote the obtained subword of $u$. The word $u_1'$ is a subword of the word $u$ and is not a subword of $v$, and we later prove that $|u_1'| \leqslant  \frac34n+4$.

In the second case we fix some occurrence of $v_1$ in $v$. There are at most $\frac{a+1}2$ different 1-overlays of $v$ on $u$ such that $v_{\short}^{-\alpha,\beta}$ coincides with $u_{\short}^{-\alpha,\beta}$ (in particular, $v_1$ coincides with an occurrence of an equal subword). For each 1-overlay there is a block $0^x$ in $v$ which is bigger than the corresponding block $0^{\alpha_u}$ in $u$. So, for each 1-overlay we can add the block $0^{\alpha_u+1}$ to $v_1$ to the corresponding place, and for one of the overlays we take the block $0^x$ instead of the block $0^{\alpha_u+1}$. Let $v_1'$ denote the obtained subword of $v$. The word $v_1'$ is a subword of the word $v$ and is not a subword of $u$, and it contains $n_1$ occurrences of 1, at most $\frac{n_{\short}-\alpha}{2}$ occurrences of 0 from small blocks, at most $\frac{a+1}{2}-1$ blocks $0^{\alpha+1}$ and one block  $0^x$. Since $\alpha+1 \leqslant x$, then $$|v_1'| \leqslant  n_1 + \frac{(n_{\short}-\alpha)}2 + \frac{(a-1)(\alpha+1)}2 + x = n_1 + \frac{n_{\short}}2 + \frac{(a-2)(\alpha+1)+1}2 + x  \leqslant $$ $$ \leqslant  n_1 + \frac{n_{\short}}2 + \frac{ax}2+\frac12 
= n_1 + \frac{n_0}{2} + \frac12 \leqslant  \frac34n + 1.$$

Since $|u_1'| = |u_1|+x$ and $|v_1'| \geqslant |u_1|+x$, we have $|u_1'| \leqslant \frac34n+1$.

\medskip

\noindent {\bf Case 3.1.2.} $u_{\longg}^{+\alpha,\beta}$ and $v_{\longg}^{+\alpha,\beta}$ are words of the third type.

\medskip

In this case $a = b$ and $n_{\short,u} = n_{\short,v} = n_{\short}$. Let $\beta_u \leqslant  \alpha_u, \alpha_v, \beta_v$. Then $n_1 = k(a+2)$. Consider the subword $u_1 = (1^{p_1}0^{q_1}\cdots 1^{p_i}0^{q_i})1^{n_1-ks}$ of $u$. There are at most $\frac{k(a+2)}{ks}  =  \frac{a+2}{2}$ ways to choose an occurrence of a subword of $u_{\short}^{-\alpha,\beta}$ which is equal to $u_1$. Then either $v_{\short}^{-\alpha,\beta}$ does not have a subword which is equal to $u_1$ or there is a subword $v_1 = u_1$ of $v_{\short}^{-\alpha,\beta}$. In the first case we add to $u_1$ a block $0^x$; and we let $u_1'$ denote the obtained subword of $u$. The word $u_1'$ is a subword of the word $u$ and is not a subword of $v$. We later prove that $|u_1'| \leqslant  \frac34n+4$. 

In the second case we fix some occurrence $v_1$ of $v$, where all 0's of this occurrence are taken from small blocks, except for $0^{\alpha_v}$ and $0^{\beta_v}$ (we do not take any 0's from these blocks). 
There are at most $\frac{a+2}2$ different 1-overlays of $v$ on $u$ such that $v_{\short}^{-\alpha,\beta}$ coincides with $u_{\short}^{-\alpha,\beta}$ (in particular, $v_1$ coincides with an occurrence of a subword equal to $u_1$). For each 1-overlay there is a block $0^x$ in $v$ which is bigger than the corresponding block $0^{\alpha_u}$ or $0^{\beta_u}$ in $u$. So, for each 1-overlay we can add the block $0^{\alpha_u+1}$ or $0^{\beta_u+1}$ to $v_1$ to the corresponding place. Notice that there is at least one 1-overlay for which we can add the block $0^{\beta_u+1}$ (which is not the case for $0^{\alpha_u+1}$). Besides that, in the resulting subword we take one block $0^x$ instead of one of the blocks $0^{\alpha_u+1}$ (or instead of one of the blocks $0^{\beta_u+1}$ if we did not add the block $0^{\alpha_u+1}$). Let $v_1'$ denote the obtained subword of $v$. The word $v_1'$ is a subword of the word $v$ and is not a subword of $u$ and it contains $n_1$ occurrences of 1, at most $\frac{n_{\short}-\alpha_v-\beta_v}{2} = \frac{n_{\short}-\alpha_u-\beta_u}{2}$ occurrences of 0 from small blocks, one block $0^{\beta+1}$, one block $0^x$ and at most $\frac{a+2}{2}-2$ blocks $0^{\alpha_u+1}$ or $0^{\beta_u}+1$. Since $\beta_u+1 \leqslant \alpha_u+1 \leqslant x$, we have
$$|v_1'| \leqslant  n_1 + \frac{(n_{\short}-\alpha_u-\beta_u)}2+x+\beta+1+ \frac{(a-2)(\alpha+1)}2 = $$ $$ = n_1 + \frac{n_{\short}}2 + \frac{2x+(\beta+1)+ (a-3)(\alpha+1)+2}2 \leqslant $$ $$\leqslant  n_1 + \frac{n_{\short}}2 + \frac{ax}2 + 1  = n_1 + \frac{n_0}{2} + 1  \leqslant  \frac34n + 1.$$

Since $|u_1'| = |u_1|+x$ and $|v_1'| \geqslant |u_1|+x$, we have $|u_1'| \leqslant \frac34n+1$.

\medskip

\noindent {\bf Case 3.2.} $s = 1$.

\medskip

First we prove the following claim. 


\begin{claim}
    Under the conditions of Case {\bf 3.2}, let $w$ be a distinguishing subword for the words $u$ and $v$. Suppose that $w$ has $n_1$ occurrences of 1, at least one block $0^x$ and at most $\frac23 ax + \frac45 (\alpha_u+\beta_u)+\frac{23}{10}$ occurrences of 0 from $t$ blocks of 0's. If $t \geqslant \frac23a$, then there is a distinguishing subword for the words $u$ and $v$ of length at most $\frac34n+4$.
\end{claim}

\begin{proof}
Since $\alpha_u+\beta_u = \alpha_v+\beta_v$,  without loss of generality we can assume that $w$ is a subword of the word $u$ and is not a subword of $v$. Notice that there are at least $\frac13 ax + \frac15 (\alpha_u+\beta_u)+(l-(a+2))-\frac{23}{10}$ occurrences of 0 which do not belong to $w$. If $|w| \leqslant  \frac34 n + 4$, then $w$ itself is a desired subword. Otherwise  $$\frac13 ax + \frac15 (\alpha_u+\beta_u)+l-a-2-\frac{23}{10} <  \frac14n-4.$$ Multiplying the inequality by 3, we get an equivalent inequality 
$$ax + \frac35 (\alpha_u+\beta_u)+3l-3a <  \frac34n+\frac{9}{10}.$$ 

Consider the word $w' = (01)^l$ and fix some occurrence of it in $u$. We add to $w'$ the same blocks of 0's as in $w$, and we let $w''$ denote the obtained subword of $u$. Since $u_{\short}^{-\alpha,\beta}$, $v_{\short}^{-\alpha,\beta}$, $u_{\longg}^{+\alpha,\beta}$, $v_{\longg}^{+\alpha,\beta}$ are periodic with period $k$, $w''$ is a subword of the word $u$ and is not a subword of $v$.  Moreover, 
$$|w''| \leqslant  |w'| +\left (\frac23 a(x-1) + \frac45 (\alpha_u+\beta_u)+\frac{23}{10}-t\right) = 2l + \frac23 a(x-1) + \frac45 (\alpha_u+\beta_u)+\frac{23}{10}-t=$$ 
$$=(ax + \frac35 (\alpha_u+\beta_u)+3l-3a) - (l-a) - \left(\frac13ax - \frac43a-\frac15(\alpha_u+\beta_u)+t\right)+\frac{23}{10} < 
\frac34n+4,$$ since $l - a \geqslant  0$, $a \geqslant  2$ and  
$$\frac13ax - \frac43a-\frac15(\alpha_u+\beta_u)+t \geqslant  \frac13ax - \frac43a-\frac25x + \frac23a = \frac13\left(a-\frac65\right)(x-2) - \frac45 \geqslant  -\frac{4}{5},$$
where the two latter inequalities follow from the inequalities $t \geqslant \frac23a$ and $\alpha_u, \beta_u \leqslant x$. The claim is proved. \qedSymC
    
\end{proof}

If $u_{\longg}^{+\alpha,\beta} = v_{\longg}^{+\alpha,\beta}$, then $u = v$ since $u_{\short}^{-\alpha,\beta} \cong v_{\short}^{-\alpha,\beta}$ and $u_{\short}^{-\alpha,\beta}$ is periodic and the length of period is $k$. A contradiction. So, $u_{\longg}^{+\alpha,\beta} \neq v_{\longg}^{+\alpha,\beta}$. Let $$u_{\longg}^{+\alpha,\beta} = 0^{\alpha_u}1^k(0^x1^k)^{r_1-1}0^{\beta_u}1^k(0^x1^k)^{s_1-1},$$ $$v_{\longg}^{+\alpha,\beta} = 0^{\alpha_v}1^k(0^x1^k)^{r_2-1}0^{\beta_v}1^k(0^x1^k)^{s_2-1}.$$ We know that $\alpha_u + \beta_u = \alpha_v + \beta_v$. Without loss of generality we assume that $\alpha_u \geqslant  \alpha_v \geqslant  \beta_v \geqslant  \beta_u$. Consider four cases.

\medskip

\noindent {\bf Case 3.2.1.} $\beta_u < \alpha_u$ and $\beta_u+1 \leqslant  \frac23x$.  

\medskip

Consider the subword $v_1$ of $v$ which contains $n_1$ occurrences of 1 and $\beta_u+1$ occurrences of 0 from each of the blocks $0^x$ and from the block $0^{\alpha_v}$. Fix some occurrence of $v_1$ in $v$. There is at most one 1-overlay of $v$ on $u$ such that small blocks overlay on small blocks and $v_1$ coincides with an occurrence of an equal subword. If such a 1-overlay exists, then there is a block $0^x$ in $v$ which overlays on a block $0^{\alpha_u}$ or $0^{\beta_u}$. We add this block $0^x$ to $v_1$. Otherwise we add any block $0^x$ to $v_1$. Let $v_1'$ denote the obtained subword of $v$. Since $\beta_u+1 \leqslant \alpha_u$ and $\beta_u+1 \leqslant  \frac23x$, the number of occurrences of 0 in $v_1'$ is $$a(\beta_u+1)+x = \left(a-\frac{3}{2}\right)(\beta_u+1)+\frac34(\beta_u+1)+\frac34(\beta_u+1)+x \leqslant$$  
$$\left(a-\frac32\right)\frac23x+\frac34(\alpha_u+\beta_u+1)+x\leqslant  \frac23 ax + \frac45(\alpha_u+\beta_u)+\frac34.$$ Note that $v_1'$ has 0's from at least $a$ blocks of 0's. So, all conditions from the claim above hold. Hence there is a distinguishing subword for the words $u$ and $v$ of length at most $\frac34n+4$.

\medskip

In the following cases we have either $\beta_u = \alpha_u$ or $\beta_u+1>\frac23x$. In both cases we have $\beta_u+1>\frac23\alpha_u-1$. Besides that, since $\alpha_u \geqslant \alpha_v \geqslant \beta_v \geqslant \beta_u$, we have $\beta_v+1>\frac23\alpha_v$.

\medskip

\noindent {\bf Case 3.2.2.} $\alpha_u>\alpha_v$. 

\medskip
Since $\alpha_u+\beta_u= \alpha_v+\beta_v$, we have $\alpha_u > \alpha_v \geqslant \beta_v > \beta_u$. We start building a distinguishing subword by taking  $u_1$ to be the subword of $u$ containing $n_1$ occurrences of 1. We index blocks of 0's in $u_{\longg}^{+\alpha,\beta}$ clockwise, and we split all blocks into groups such that each group contains $t$ blocks with indices $i, i+r_2, i+2r_2, \ldots, i+(t-1)r_2$ for some $i$. We add $\alpha_v+1$ zeros to $u_1$ from every other 
block from each group (either from blocks with indices $i, i+2r_2, i+4r_2, \ldots$ or from blocks with indices $i+r_2, i+3r_2, i+5r_2, \ldots$). It is possible since there is only one block of 0's in $u$ which is smaller than $0^{\alpha_v+1}$ (this is a block $0^{\beta_u}$).  We added $\lceil \frac{t}{2} \rceil \leqslant  \frac23$ blocks of 0's from each group, that is, at most  $\frac{2(a+2)}3$ blocks of $0^{\alpha_v+1}$. We also replaced one block $0^{\alpha_v+1}$ with  $0^x$. Without loss of generality, we can assume that we added 0's exactly from $\frac{2(a+2)}3$ blocks. Let $u_1'$ denote the obtained subword of $u$. The word $u_1'$ is a subword of the word $u$ and is not a subword of $v$. Moreover, since $\alpha_v+1 \leqslant x$ and $\beta_v+1 \geqslant \frac23\alpha_v$, the number of 0's in $u_1'$ has at most $$\left(\frac{2(a+2)}3-1\right)(\alpha_v+1)+x \leqslant  \frac23 ax + \frac43 (\alpha_v+1) \leqslant  \frac23 ax + \frac45 (\alpha_v+\beta_v+1)+\frac43.$$ The last inequality holds since $$\frac45 (\alpha_v+\beta_v+1) \geqslant  \frac45\alpha_v+\frac45\cdot\frac23\alpha_v = \frac43 \alpha_v.$$
Since $\frac45+\frac43 < 2.3$, all conditions from the above claim hold. Hence, there is a distinguishing subword for the words $u$ and $v$ of length at most $\frac34n+4$.

\medskip

\noindent {\bf Case 3.2.3.} $\alpha_u=\alpha_v=\beta_u=\beta_v$. 

\medskip

If $\alpha_u = \alpha_v = x$, then $u = v$, a contradiction. Then $\alpha_u=\alpha_v \leqslant  x-1$. 

Since $u \neq v$, we have $r_1 \neq r_2,s_2$ and $s_1 \neq r_2, s_2$. We index blocks of 0's in $u_{\longg}^{+\alpha,\beta}$ and split all blocks into groups as in Case {\bf 3.2.2} in a way that each group contains $t$ blocks with indices $i, i+r_2, i+2r_2, \ldots, i+(t-1)r_2$ for some $i$.  
If the blocks $0^{\alpha_u}$ and $0^{\beta_u}$ are in the different groups, then we can add $\alpha_v+1$ zeros to $u_1$ from every second block from each group and
proceed with the proof as in Case {\bf 3.2.2}. Assume that blocks $0^{\alpha_u}$ and $0^{\beta_u}$ are in the same group, and the blocks $0^{\alpha_u}$ and $0^{\beta_u}$ have indices $i$ and $i+hr_2$, respectively.

Since $\{r_1,s_1\} \neq \{r_2,s_2\}$,  $0^{\alpha_u}$ and $0^{\beta_u}$ are not consecutive blocks in their group. That is, $h \neq 1$ and $h \neq t-1$. In particular, $t \geqslant  4$. From each group which does not contain the blocks $0^{\alpha_u}$ and $0^{\beta_u}$, we can add $0^x$ to $u_1$ from every second block (from blocks with indices $i, i+2r_2, i+4r_2, \ldots$). If $h$ is even, then for the group containing the blocks $0^{\alpha_u}$ and $0^{\beta_u}$ we can add $0^x$ to $u_1$ from blocks with indices $i+r_2, i+3r_2, i+5r_2, \ldots$. If $h$ is odd, then for the group containing the blocks $0^{\alpha_u}$ and $0^{\beta_u}$ we can add $0^x$ to $u_1$ from blocks with indices $i+r_2, i+3r_2, \ldots, i+(h-2)r_2, i+(h-1)r_2, i+(h+1)r_2, i+(h+3)r_2, i+(h+5)r_2, \ldots$. It is not difficult to verify that for any $t \geqslant  4$ we add at most $\frac23(a+2)$ blocks $0^x$ to $u_1$. We let $u_1'$ denote the obtained subword of $u$. The proof in this case can be completed similarly to Case {\bf 3.2.2}.

\medskip

\noindent {\bf Case 3.2.4.} $\alpha_u=\alpha_v > \beta_u = \beta_v$.  

\medskip

If $\alpha_u = \alpha_v = x$, then $u = v$, a contradiction. Then $\alpha_u=\alpha_v \leqslant  x-1$.

We now index blocks of 0's in $u_{\longg}^{+\alpha,\beta}$ clockwise and split all blocks into groups similarly to Case {\bf 3.2.2}: each group contains $t$ blocks with indices $i, i+r_2, i+2r_2, \ldots, i+(t-1)r_2$ for some $i$. If the blocks $0^{\alpha_u}$ and $0^{\beta_u}$ are in different groups, then we can proceed with the proof as in Case {\bf 3.2.2}. Assume that blocks $0^{\alpha_u}$ and $0^{\beta_u}$ are in the same group, and the blocks $0^{\alpha_u}$ and $0^{\beta_u}$ have indices $i$ and $i+hr_2$, respectively. If $h \neq t-1$ then we can proceed with the proof as in Case {\bf 3.2.3}. Assume $h = t-1$.

From each group which does not contain the blocks $0^{\alpha_u}$ and $0^{\beta_u}$ we can add $0^x$ to $u_1$ from every second block (from blocks with indices $i, i+2r_2, i+4r_2, \ldots$). If $t$ is odd, then from the group containing the blocks $0^{\alpha_u}$ and $0^{\beta_u}$ we can add to $u_1$ block $0^{\alpha_u}$ and blocks $0^x$ from blocks with indices $i+r_2, i+3r_2, \ldots, i+(t-2)r_2$. If $t$ is even, then for the group containing the blocks $0^{\alpha_u}$ and $0^{\beta_u}$ we can add to $u_1$ the block $0^{\alpha_u}$ and the blocks $0^x$ from blocks with indices $i+r_2, i+3r_2, \ldots, i+(t-3)r_2, i+(t-2)r_2$. If $t \neq 4$ then it is not difficult to verify that we add at most $\frac23t$ blocks to $u_1$ from each group. If $t = 4$, then there are at least two groups of blocks since $a \geqslant 3$ (there are at least three big blocks in $u$). So from each group which does not contain the blocks $0^{\alpha_u}$ and $0^{\beta_u}$ we add $\frac{t}2 = 2$ blocks to $u_1$ and from the group containing the blocks $0^{\alpha_u}$ and $0^{\beta_u}$ we add $3$ blocks. So we add at most $\frac23(a+2)$ blocks $0^x$ to $u_1$. We let $u_1'$ denote the subword obtained from $u$. The proof in this case can be completed similarly to Case {\bf 3.2.2}. \qed

 \section{Conclusions}
In this paper, we provided lower and upper bounds for the minimal length $k$ which is sufficient to distinguish two cyclic words of length $n$ by sets of their subwords of length $k$. The lower bound is given in Proposition~\ref{pr:lower_bound}, and the upper bound is provided in Theorem~\ref{th:main}; the difference between lower and upper bounds is bounded by a constant which is at most $5$. We note that comparing the lower bound from Proposition~\ref{pr:lower_bound} with computational results given in Table 1, one can notice that the bound from the proposition is likely to be optimal starting from some length, except for the values $n=4m+6$, where it is smaller by 1. However, the examples giving a better bound do not seem to be generalizable for bigger values of $n$, so it is likely that Proposition~\ref{pr:lower_bound} gives the optimal length. Concerning the upper bound, the proof of Theorem~\ref{th:main} can probably be pushed to reduce the upper bound (with more technical details). However, it is not clear if it could be pushed further to get the precise value of $k$. So, establishing the exact value of $k$ is still an open question. Another open problem is finding the length which allows to recover cyclic words from sets of their factors of length $k$ with multiplicites.

 \section*{Acknowledgements}

This work was supported by the Russian Science Foundation, project 23-11-00133.

\end{document}